\documentclass[12pt]{article}
\usepackage{amsmath,amssymb}
\usepackage{amsthm}
\usepackage{graphicx}
\usepackage{url}
\usepackage{multirow}
\usepackage[font=footnotesize]{subfig}
\usepackage{array}
\usepackage{color}
\usepackage[round]{natbib}

\newtheorem{theorem2}{Theorem}

\newtheorem{lemma}{Lemma}

\begin{document}

\newcommand\colred{}

\begin{center}
{\bf\large Testing for Heteroscedasticity in High-dimensional Regressions}\\[2mm]
Zhaoyuan Li and Jianfeng Yao\\
Department of Statistics and Actuarial Science\\
The University of Hong Kong\\
\end{center}

\begin{abstract}
Testing heteroscedasticity of the errors is a major challenge in high-dimensional regressions where the number of covariates  is large compared to the sample size. Traditional procedures  such as the  White and the Breusch-Pagan tests typically suffer from low sizes and powers. This paper proposes two new test procedures based on standard OLS residuals.  Using the theory of random Haar orthogonal matrices, the asymptotic normality of both test statistics is obtained under the null when the degrees of freedom tend to infinity. This encompasses both the classical low-dimensional setting where the number of variables is fixed while the sample size tends to infinity, and the proportional high-dimensional setting where these dimensions grow to infinity proportionally.  These procedures thus offer a wide coverage of dimensions in applications.   To our best knowledge, this is the first procedures in the literature for testing heteroscedasticity which are valid for medium and high-dimensional regressions. The superiority of our proposed tests over the existing methods are demonstrated by extensive simulations and by several real data analyses as well. 

\end{abstract}

\section{Introduction}
Consider the linear regression model
\begin{eqnarray}\label{model}
y_i=X_i\boldsymbol{\beta} +\varepsilon_i,  \quad i=1,\ldots, n,
\end{eqnarray}
where $y_i$ is the dependent variable, $X_i$ is a $1\times p$ vector of regressors, $\boldsymbol{\beta}$ is the $p$-dimensional coefficient vector, and {\colred{} the error $\varepsilon_i = \sigma_i \eta_i$, where $\sigma_i$ could depend on covariates $X_i$, and $\{\eta_i\}$ are independent standard normal distributed.} A significant part of the inference theory for the model is based on the assumption that the errors $\{\varepsilon_i\}$ are homoscedastic, i.e. under the hypothesis
\begin{eqnarray}\label{hypo}
H_0: \sigma_1^2=\ldots =\sigma_n^2=\sigma^2,
\end{eqnarray}
for some constant $\sigma^2 >0$, { \colred{} that is, the unconditional and conditional variances of the noise coincide and are independent of the covariates.} However, this assumption cannot be always guaranteed in practice, and it is well known that heteroscedasticity of the error variance leads to inefficient parameter estimates and inconsistent covariance estimates. We consider testing the hypothesis in (\ref{hypo}) when { \colred{} the number of covariates $p$ goes to infinity together with the sample size $n$.}

Studying this testing problem is also motivated by recent advances in the estimation of high-dimensional regressions. In this paper, we consider testing the hypothesis in (2) when the number of covariates $p$ can be large with respect to the sample size $n$. High-dimensional regressions become vital due to the increasingly wide availability of data sets with a large number of variables in empirical economics, finance \citep{belloni2014high} and biology \citep{daye2012high}. For example, the American Housing Survey records prices as well as a multitude of features of the house sold; scanner datasets record prices and numerous characteristics of products sold at a store or on the internet \citep{belloni2011inference}; and text data frequently lead to counts of words in documents from a large dictionary \citep{taddy2013multinomial}. Importantly, heteroscedasticity is possible in such data sets. Among the existing methods for high-dimensional regressions, one approach assumes sparsity where the number of important regressors is much smaller than $n$ or $p$. For example, \cite{belloni2012sparse} and \cite{belloni2014pivotal} studied the estimation problem with heteroscedasticity and proposed the heteroscedastic form of Lasso and square-root Lasso methods, respectively. However, if the errors are homoscedastic, these heteroscedasticity-consistent methods may lose efficiency as suggested by the phenomenon arising in low-dimensional regressions. Here, we conduct a small simulation study with 5000 replications to illustrate this point by using data generated according to Model 1 in Section 3 with $p=100, n=250, \boldsymbol{\beta} =(\mathbf{1}_{50}^\prime,\mathbf{0}_{50}^\prime)^\prime$. {\colred{}The ratio of the heteroscedastic form of Lasso estimator (root mean squared error) over the OLS estimator (root mean squared error)} is 0.8791 when the errors are related to 10 regressors. But the ratio is 12.08 for homoscedastic errors, and 7.53 when the errors are related to only one regressor.   \cite{el2013robust} and \cite{bean2013optimal} stated that the Lasso-type of methods result in biased estimates of the coefficients, and the least squares method is preferable to other M-estimators in high-dimensional regression under homoscedasticity. \cite{bean2013optimal} proposed an optimal least square algorithm with the assumption that the error is homoscedastic with a known distribution. However, the performance of this optimal algorithm is largely unknown if the error is in fact heteroscedastic. In summary,  the discussion above on two recent high-dimensional estimation methodologies highlights the importance of conducting heteroscedasticity detection as a preliminary step in practice in order to select a suitable estimation method for high-dimensional regressions.

Heteroscedastic testing has been extensively studied for classical low-dimensional regressions in the literature. Many popular tests examine whether the estimated residuals are correlated with some covariates or any auxiliary variables that would be useful in explaining the departure from homoscedasticity, see for example \citet{breusch1979simple}, \citet{white1980heteroskedasticity}, \citet{cook1983diagnostics}, \citet{azzalini1993use}, \citet{diblasi1997testing}, and \citet{su2013}.
These tests, however, will not have much power if the existing heteroscedasticity is not strongly related to either the chosen auxiliary variables or covariates. In consequence, many nonparametric test procedures are thus proposed to avoid such potential model misspecification, see for example, \citet{eubank1993detecting} and \citet{dette1998testing}. \citet{koenker1982robust} and \citet{newey1987asymmetric} proposed to test heteroscedasticity by comparing different quantile or expectile estimates. Their approach is much preferable to many other tests for heavy tailed errors \citep{lee1992heteroskedasticity}. However, there is some difficulty in applying this approach because no clear criterion exists for selecting the used quantiles.

Testing the homoscedasticity hypothesis (2) becomes very challenging for high-dimensional regressions. The large sample theory of all the existing tests discussed above is developed under the {\it low-dimensional} framework where the dimension $p$ should be fixed while the sample size tends to infinity. By referring to recent advances in high-dimensional statistics \citep{paul2014random, yao2015large}, it clearly appears that these test methods are not suitable for analysing data sets where the number of variables is not $``$small enough" compared to the sample size. For example, the limiting $\chi_{p(p+1)/2}^2$ approximation for White's test statistic is typically misleading even for a moderate dimension $p=25$ while the sample size is $n=500$ (see Table 1 for more details). As an additional illustration, many published Monte Carlo studies of tests for heteroscedasticity have used very low-dimensional designs and the error variances are determined by a single variable in the alternative model, see for example \citet{dette1998testing}. \citet{godfrey1999robustness} and \citet{godfrey1996some} showed that the results obtained from very simple experimental designs (for example $p=1$) may be an unreliable guide to finite sample performance with a moderately large number of variables. Another illustration of high-dimensional effect is an interesting phenomenon shown in \citet{ferrari2002corrected} and \citet{godfrey1999robustness} where the actual size of many popular tests stays far from the nominal level for the moderately large sample size $n$. Therefore, accurate and powerful test procedure is an urgent need for detecting heteroscedasticity in a high-dimensional regression.

In this paper, we propose two new procedures for testing heteroscedasticity, which are {\it dimension-proof} in the sense that they are valid for a wide range of dimension (covering both low and high-dimensional settings). More precisely, our procedures are theoretically valid once the degree of freedom $n-p$ is large enough (precisely when $n-p\to \infty$). This includes for instance the low-dimensional setting where $p\ll n$ and the high-dimensional situation where $p$ and $n$ grow to infinity proportionally such that $p\propto cn$ with $0<c<1$. Simulation experiments reported in Section 3 show that the proposed tests outperform the popular existing methods for medium or high-dimensional regressions. More surprisingly, even in low-dimensional setting, our procedures perform better than these classical procedures.

The paper is organised as follows. The main results of the paper are reported in Section 2. Two new tests are here proposed using the residuals of a least squares fit. Section 3 reports several simulation experiments to assess the finite sample performance of the proposed tests and compare them to the existing ones. In Section 4 we apply the suggested procedures to analyse four real data sets. All technical proofs of the results presented in Section 2 are relegated to the Appendix.

\section{Main results}
The following assumptions will be used in our set-up of the regression model (\ref{model}):
\begin{itemize}
\item Assumption (a): The errors are independent and normal distributed: $\varepsilon_i \sim N(0, \sigma_i^2), i=1, \ldots, n$;
\item {\colred{} Assumption (b): In
  the $n\times p$ design matrix $\mathbf{X}=(X_1^\prime, \ldots,
  X_n^\prime)'$, $\{X_i\}_{1\leq i \leq n}$ are independent normal distributed vectors $N(\mathbf{0, \Sigma})$ with mean 0 and covariance matrix $\mathbf{\Sigma}$;}
\item Assumption (c): As $n\to \infty$, the degree of freedom $k=k(n) :=n-p \to \infty$;
\item Assumption (d): In addition to Assumption (c), $\displaystyle \liminf_{k, n} c_n >0$, where $c_n =\frac{k}{n}$.
\end{itemize}
Both Assumptions (a) and (b) are classical in a regression model. Assumptions (c) and (d) define the asymptotic setting of the paper which is quite general. In particular, the setting includes the situation where both $p$ and $n$ are large while remaining comparable, i.e. for some $0<c<1$, $p \simeq c\cdot n$ and $k \simeq (1-c)\cdot n$. Meanwhile, the setting encompasses the classical low-dimensional situation where $p$ is a constant and $n\to \infty$. {\colred{} Therefore, the procedure derived under this setting will be applicable to both the high and low-dimensional settings. It is however noted that since our methods will use the OLS residuals, it is required that $p<n$ although both dimensions can grow to infinity.}

In the regression model (\ref{model}) and under homoscedasticity, the parameter vector $\boldsymbol{\beta}$  is estimated by the OLS estimator $\widehat{\boldsymbol{\beta}}_0=(\mathbf{X}^\prime \mathbf{X})^{-1} \mathbf{X}^\prime Y$ where $Y=(y_1, \ldots, y_n)^\prime$ and $\mathbf{X}= (X_1^\prime, \ldots, X_n^\prime)^\prime$.
Then, the vector of residuals is
\begin{eqnarray}\label{error_h1}
\hat{\boldsymbol{\varepsilon}} = Y-\mathbf{X}\widehat{\boldsymbol{\beta}}_0=\mathbf{Q}_x \boldsymbol{\varepsilon}, \quad \textrm{with}\ \mathbf{Q}_x=\mathbf{I}_n - \mathbf{X}(\mathbf{X}'\mathbf{X})^{-1}\mathbf{X}^\prime.
\end{eqnarray}
Here and throughout of the paper, $\mathbf{I}_n$ denotes the $n$-th order identity matrix. Notice that $\mathbf{Q}_x$ is a projection matrix of rank $k=n-p$.
In the following, two test statistics are proposed based on the residuals $\hat{\boldsymbol{\varepsilon}}=\{\hat{\varepsilon}_i\}$.

{\colred{}Note that each covariate vector $X_i\sim N(\mathbf{0, \Sigma})$ so we have $X_i = \mathbf{\Sigma}^{1/2} Z_i$ where $Z_i\sim N(\mathbf{0,I}_p)$. Let $\mathbf{Z} = (Z'_1, \ldots, Z'_n)$ be the corresponding $``$design" matrix. Then we have $\mathbf{X}(\mathbf{X}'\mathbf{X})^{-1}\mathbf{X}' = \mathbf{Z}(\mathbf{Z}'\mathbf{Z})^{-1}\mathbf{Z}'$. Therefore the projection matrix $\mathbf{Q}_x$ is independent of the covariance structure $\mathbf{\Sigma}$. In what follows we can assume $\mathbf{\Sigma} = \mathbf{I}_p$ and the $p$ coordinates of $X_i$ are i.i.d standard normals. }

\subsection{An approximate likelihood-ratio test}
We first derive a test statistic from the concept of likelihood ratio test. For the regression model (\ref{model}) and under Assumption (a), the likelihood function is simply
\begin{eqnarray*}
L(\boldsymbol{\beta}, \sigma_1^2, \ldots, \sigma_n^2)=(2\pi)^{-n/2} \left(\sigma_1^2\cdots\sigma_n^2\right)^{-1/2} \exp \left\{-\frac{1}{2}\sum_{i=1}^n \frac{(y_i-X_i\boldsymbol{\beta})^2}{\sigma_i^2}\right\}.
\end{eqnarray*}
Without assuming the homoscedasticity, the likelihood is maximised by solving the system of equations
\begin{eqnarray*}
\left\{ \begin{array}{l} \displaystyle\frac{\partial \log L}{\partial \sigma_i^2}=-\frac{1}{2\sigma_i^2} +\frac{1}{2\sigma_i^4}(y_i-X_i \boldsymbol{\beta})^2 =0, \\
\displaystyle\frac{\partial \log L}{\partial \boldsymbol{\beta}} =-\frac{1}{2} \sum_{i=1}^n\frac{2(y_i-X_i\boldsymbol{\beta})}{\sigma_i^2}(-X_i)=0, \end{array}\right. \quad 1\leq i\leq n.
\end{eqnarray*}
Therefore, the maximum likelihood estimator (MLE) $(\hat{\boldsymbol{\beta}}, \hat{\sigma}_1^2, \ldots, \hat{\sigma}_n^2)$ of $(\boldsymbol{\beta}, \sigma_1^2, \ldots, \sigma_n^2)$ satisfy the system of equations
\begin{eqnarray*}
\left\{ \begin{array}{l} \displaystyle\hat{\sigma}_i^2=(y_i-X_i \hat{\boldsymbol{\beta}})^2, \\
\displaystyle \hat{\boldsymbol{\beta}}=\left( \sum_{i=1}^n \frac{X_i' X_i}{\hat{\sigma}_i^2} \right)^{-1} \sum_{i=1}^n \frac{y_iX_i'}{\hat{\sigma}_i^2}, \end{array}\right. \quad 1 \leq i \leq n.
\end{eqnarray*}
The corresponding maximized likelihood is
\begin{eqnarray*}
L_1=(2\pi)^{-n/2}  \prod \{(y_i-X_i \hat{\boldsymbol{\beta}})^2\}^{-1/2} \exp(-n/2).
\end{eqnarray*}
Notice that since the number of unknown parameters $p+n$ exceeds the sample size, this MLE cannot be a reliable estimator. Nevertheless, this likelihood concept will help us to define a meaningful test statistic for testing the homoscedasticity hypothesis as follows: we approximate the MLE $\hat{\boldsymbol{\beta}}$ in the maximized likelihood $L_1$ by the OLS $\hat{\boldsymbol{\beta}}_0$ to get an approximate value
\begin{eqnarray*}
L_1^\ast = (2\pi)^{-n/2} \prod \{(y_i -X_i \hat{\boldsymbol{\beta}}_0)^2\}^{-1/2} \exp(-n/2).
\end{eqnarray*}
On the other hand under the homoscedasticity hypothesis, the OLS estimator $\hat{\boldsymbol{\beta}}_0$ and the estimator of the variance
\[\hat{\sigma}_0^2=\frac{1}{n}\sum_{i=1}^n (y_i-X_i\hat{\boldsymbol{\beta}}_0)^2,\]
are in fact the MLEs.
So the maximized likelihood under the null hypothesis is
\begin{eqnarray}
L_0=(2\pi)^{-n/2}(\hat{\sigma}_0^2)^{-n/2}\exp(-n/2).
\end{eqnarray}
Therefore, the approximate likelihood ratio, {\colred{}likelihood ratio is first derived by \citet{mauchly1940significance}}, is defined as
\begin{eqnarray*}
\frac{L_0}{L_1^\ast}=\frac{(\hat{\sigma}_0^2)^{-n/2}}{\left( \prod_{i=1}^n (y_i-X_i\hat{\boldsymbol{\beta}}_0)^2\right)^{-1/2}} = \left\{ \frac{\frac{1}{n}\sum_{i=1}^n \hat{\varepsilon}_i^2}{\left( \prod_{i=1}^n \hat{\varepsilon}_i^2 \right)^{1/n}} \right\}^{-\frac{n}{2}},
\end{eqnarray*}
where it is reminded that $\hat{\varepsilon}_i=Y_i-X_i\hat{\boldsymbol{\beta}}_0$. This suggests to consider the \emph{approximate likelihood-ratio statistic}
\begin{eqnarray}
T_1= -\frac{2}{n}\log \frac{L_0}{L_1^\ast}= \log \frac{\frac{1}{n}\sum_{i=1}^n \hat{\varepsilon}_i^2}{\left( \prod_{i=1}^n \hat{\varepsilon}_i^2 \right)^{1/n}}.
\end{eqnarray}
Interestingly enough, the statistic $T_1$ depends on the ratio of the arithmetic mean of the squared residuals over their geometric mean: $T_1\geq 0$ always and a large value of $T_1$ will indicate a significant deviation of the residuals $\{\hat{\varepsilon}_i^2\}$ from a constant, that is presence of heteroscedasticity. Meanwhile, this statistic has a scale-free property and is not affected by the magnitude of the variance $\sigma^2$ under the null hypothesis. Therefore, without loss of generality for the study of $T_1$, we assume that $\sigma^2=1$ under the null. The asymptotic distribution of $T_1$ under the null is derived in the following theorem.
\begin{theorem2}\label{T1}
Assume that Assumptions (a)-(b)-(d) are satisfied for the regression model (\ref{model}). Then under the null hypothesis of homoscedasticity, we have as $n\to \infty$
\begin{eqnarray}\label{T1_normality}
\sqrt{n}\left(T_1-[\log 2 +\gamma]\right)\overset{\mathcal{D}}{\longrightarrow} \mathcal{N} \left(0, \frac{\pi^2}{2}-2 \right),
\end{eqnarray}
where $\gamma\approx 0.5772$ is the Euler constant.
\end{theorem2}
The testing procedure using $T_1$ with the critical value from (\ref{T1_normality}) is referred as the \emph{approximate likelihood-ratio test} (ALRT). In addition to the scale-free property mentioned above, an attractive feature appears here is that the asymptotic distribution of $T_1$ is completely independent of $p/n$, the relative magnitude of the dimension $p$ over the sample size $n$. This prefigures a large applicability of the procedure to a wide range of combinations of $(p, n)$ in finite-sample situations. This robustness is indeed confirmed by the simulation study reported in Section 3.

The proof of Theorem \ref{T1} is based on the following lemma, which establishes the asymptotic limit of the joint distribution of  $\sum_{i=1}^n \hat{\varepsilon}_i^2$ and $\sum_{i=1}^n \log \hat{\varepsilon}_i^2$ under the null.
\begin{lemma}\label{joint1}
Let $\{\hat{\varepsilon}_i\}_{1\leq i \leq n}$ be the sequence of the OLS residuals given in (\ref{error_h1}). Then, under $H_0$ and Assumptions (a)-(b)-(d), and as $n\to \infty$, we have
\begin{eqnarray}
\mathbf{\Sigma}_1^{-1/2}\left\{\left(\begin{array}{c} \sum_{i=1}^n \hat{\varepsilon}_i^2 \\ \sum_{i=1}^n \log \hat{\varepsilon}_i^2 \end{array}\right) -\boldsymbol{\mu}_1\right\} \overset{\mathcal{D}}{\longrightarrow} \mathcal{N}(\mathbf{0}, \mathbf{I}_2),
\end{eqnarray}
where 
\begin{eqnarray*}
\boldsymbol{\mu}_1 = \left(\begin{array}{c} k \\ n\left( -\gamma -\log 2 +\log c_n \right)\end{array}\right),
\end{eqnarray*}
and
\begin{eqnarray*}
\mathbf{\Sigma}_1 = \left(\begin{array}{cl} 2k & 2n \\ 2n & n\left( \pi^2/2+2/c_n-2\right) \end{array}\right).
\end{eqnarray*}
\end{lemma}
The proofs of Lemma \ref{joint1} and Theorem \ref{T1} are postponed to the appendix.

\subsection{ The coefficient-of-variation test}
The departure of a sequence of numbers from a constant can also be efficiently assessed by its coefficient of variation. In multivariate analysis, this idea is closely related to optimal invariant tests, see \citet{john1971some}. Applying this idea to the sequence of residuals $\{\hat{\varepsilon}_i\}$ leads to the following coefficient-of-variation statistic
\begin{eqnarray}
T_2=\frac{\frac{1}{n}\sum_{i=1}^n( \hat{\varepsilon}_i^2-\bar{m})^2}{\bar{m}^2}, \quad \textrm{with} \ \bar{m}=\frac{1}{n}\sum_{i=1}^n \hat{\varepsilon}_i^2.
\end{eqnarray}
Obviously, the statistic $T_2$ becomes small and close to 0 under the null hypothesis of homoscedasticity, and larger under the alternative hypothesis of heteroscedasticity. Like the previous statistic $T_1$, this statistic is also scale-free and again we can assume $\sigma^2=1$ for $T_2$ under the null without loss of generality. The asymptotic distribution of $T_2$ under the null hypothesis is derived in the following theorem.
\begin{theorem2}\label{T2}
Assume that Assumptions (a)-(b)-(c) are satisfied for the regression model (\ref{model}). Then under the null hypothesis of homoscedasticity, we have as $n\to \infty$
\begin{eqnarray}\label{T2_normality}
\sqrt{n}(T_2-2)\overset{\mathcal{D}}{\longrightarrow}N(0, 24).
\end{eqnarray}
\end{theorem2}
The testing procedure using $T_2$ with the critical value from (\ref{T2_normality}) is referred as the \emph{coefficient-of-variation test} (CVT). Similar to the statistic $T_1$, the asymptotic distribution of $T_2$ is also scale free and independent of $p/n$, the relative magnitude of the dimension $p$ over the sample size $n$.

The proof of Theorem \ref{T2} is based on the following lemma, which establishes the asymptotic limit of the joint distribution of $\sum_{i=1}^n \hat{\varepsilon}_i^4$ and $\sum_{i=1}^n \hat{\varepsilon}_i^2$ under the null.
\begin{lemma}\label{joint2}
Let $\{\hat{\varepsilon}_i\}_{1\leq i \leq n}$ be the sequence of the OLS residuals given in (\ref{error_h1}). Then, under $H_0$ and Assumptions (a)-(b)-(c), and as $n\to \infty$, we have
\begin{eqnarray}
\mathbf{\Sigma}_2^{-1/2}\left\{\left(\begin{array}{c}\sum_{i=1}^n \hat{\varepsilon}_i^4\\ \sum_{i=1}^n \hat{\varepsilon}_i^2 \end{array}\right)-\boldsymbol{\mu}_2\right\} \overset{\mathcal{D}}{\longrightarrow} \mathcal{N}(\mathbf{0}, \mathbf{I}_2),
\end{eqnarray}
where
\begin{eqnarray*}
\boldsymbol{\mu}_2=\left(\begin{array}{c} \frac{3k(k+2)}{n+2} \\ k \end{array}\right),
\end{eqnarray*}
and
\begin{eqnarray*}
\boldsymbol{\Sigma}_2=\left(\begin{array}{rc} \frac{24k^4}{n^3}+\frac{72k^3}{n^2} & \frac{12k^2}{n+2}\\ \frac{12k^2}{n+2} & 2k \end{array} \right).
\end{eqnarray*}
\end{lemma}
The proofs of Lemma \ref{joint2} and Theorem \ref{T2} are postponed in the appendix.

\section{Simulation experiments}

We have undertaken an extensive simulation study to investigate the finite sample performance of the proposed tests, ALRT and CVT. Comparisons are also made with several existing popular methods: the BP test, proposed by \citet{breusch1979simple} and modified by \citet{koenker1981note}; the White test \citep{white1980heteroskedasticity}; and the DM test \citep{dette1998testing}.

\citet{breusch1979simple} constructed a general test statistic, assuming that the conditional variance has a known functional form $h(z_t^\prime \alpha)$, where $z_t=(1, X_i)^\prime$ and $\alpha=(\alpha_0, \alpha_1, \ldots, \alpha_p)^\prime$. They proposed a Lagrange multiplier statistic to test the joint null hypothesis of $\alpha_1=\alpha_2=\cdots =\alpha_p=0$ while the intercept $\alpha_0$ is unspecified. \citet{koenker1981note} modified this test in order to improve its empirical size. This test has been widely used in the literature and is the representative one in the family of Lagrange multiplier or score tests, as it includes many other tests (e.g. \citealp{cook1983diagnostics} and \citealp{eubank1993detecting}) as special cases.

The White test fits an artificial regression of the squared OLS residuals $(\hat{\varepsilon}_i^2)$ on the elements $(x_{ij}x_{ik}, k\geq j)$ of the lower triangle of the matrix $X_i^\prime X_i$, and the test statistic is the squared multiple correlation coefficient from this regression. The author proved that the statistic is asymptotically distributed as $\chi^2$ with $p(p+1)/2$ degrees of freedom under the null hypothesis of homoscedasticity (as the sample size tends to infinity). 

\citet{dette1998testing} proposed a nonparametric method, the DM test. It is constructed on estimation of empirical variance of expected squared residuals, and its asymptotic normality is given. This nonparametric test avoids the estimation of the regression curve directly, which makes it more robust and better than those tests based the estimated residuals.

\subsection{Empirical sizes of the tests}

We explore the performance of these tests using different combination of $p$ and $n$. The sample sizes $n=100, 500, 1000$ and ratios $p/n = 0.05, 0.1, 0.3, 0.5, 0.7, 0.9$ are considered. {\colred{}Each simulation is repeated 5000 times to test the stability of the method. Empirical size of a test is the percentage of rejected tested cases.} According to the model (\ref{model}), the design matrix $X_i$ are assumed to be multi-normal. The error $\varepsilon_i$ is drawn from standard normal as the size and power of the proposed tests are invariant with respect to different scalings of variance function.  The nominal test level is 5\%.

\begin{table}
\centering
\footnotesize
\caption{Empirical sizes of the ALRT, CVT, White and BP tests with sample size $n=100, 500, 1000$ and varying ratio $p/n$ (in \%).}
\label{table1}
\begin{tabular}{|c|cccc|cccc|cccc|}
\hline
\multirow{2}{*}{$p/n$} & \multicolumn{4}{c|}{$n=100$} & \multicolumn{4}{c|}{n=500} & \multicolumn{4}{c|}{n=1000}\\
  & ALRT & CVT & White & BP & ALRT & CVT & White & BP & ALRT & CVT & White & BP\\
\hline
0.05 &4.62 & 4.24 & 5.92 & 3.74 & 4.88 & 5.66 & 0.16 & 4.64 & 5.16 & 5.48 & NA & 4.16 \\
0.1 &5.00 & 4.64 & 0.28 & 3.80 & 4.86 & 5.78 & NA & 3.64 &   5.44 & 5.20 & NA & 3.60\\
0.3 &4.98 & 4.70 & NA & 1.66 &   4.60 & 6.06 & NA & 1.88 &   5.38 & 5.06 & NA & 2.32 \\
0.5 &5.30 & 4.72 & NA & 0.52 &   4.84 & 4.80 & NA & 0.70 &   5.04 & 5.14 & NA & 0.72\\
0.7 &4.66 & 4.58 & NA & 0 &   5.60 & 5.50 & NA & 0.02 &   5.38 & 5.70 & NA & 0.02\\
0.9 &5.06 & 4.28 & NA & 0 &   5.48 & 5.24 & NA & 0  &       4.44 & 5.04 & NA & 0\\
\hline
\end{tabular}
{\em $^{*}$  \footnotesize{NA denotes $``$Not Applicable"} \hspace{9.5cm}}
\end{table}

Table \ref{table1} presents the empirical sizes of the ALRT, CVT, White and BP tests {\colred{}(values close to 5\% are better)}. The proposed ALRT and CVT tests are consistently accurate in all tested combinations of $(p, n)$ (including the smallest ones); they largely outperform the White and BP tests. This good performance can be explained by a fast convergence in the limiting results of ALRT (Theorem 1) and CVT (Theorem 2). The ALRT test performs a little better than the CVT test for small value of the ratio $p/n$, but the CVT test is preferred when $p/n$ is getting close to 1. The BP test loses its size from (approximately) 4\% to 1\% when the ratio $p/n$ increases from 0.05 to 0.5, while the White test has an empirical size of 0.16\% when the ratio is $p/n=0.05$ and sample size is $n=500$ (Notice that this test is not applicable when $p> 25$ due to its dimension-sample-size requirement $p(p+1)/2<n$).

\subsection{Empirical powers of the tests}

To investigate the power of these tests, we follow \citet{dette1998testing} and consider the following three models with different error forms:
\begin{itemize}
\item Model 1: $y_i=X_i \boldsymbol{\beta} + \varepsilon_i \exp( \mathbf{c}X_i)$;
\item Model 2: $y_i=X_i \boldsymbol{\beta} + \varepsilon_i (1+\mathbf{c}\sin(10 X_i))^2$;
\item Model 3: $y_i=X_i \boldsymbol{\beta} + \varepsilon_i (1+\mathbf{c}X_i)^2$;
\end{itemize}
where the vector $\mathbf{c}$ is filled with elements $0$ and/or $c_0=0.5$. The value $\mathbf{c}=\mathbf{0}$ corresponds to homoscedasticity, and we consider two levels of heteroscedasticity: $\mathbf{c}=(c_0\mathbf{1}_{p_0}^\prime, \mathbf{0}_{p-p_0}^\prime)^\prime$ with $p_0=1$ (1st component only) and $p_0=0.1 p$ (first 10\% of components). Same setting with Section 3.1 is used and empirical powers of the tests are obtained using 5000 replications for each scenario.

Tables \ref{table2}-\ref{table4} present the empirical powers of the
ALRT, CVT and BP tests for these three error models,
respectively. Plots are also provided for the case of sample size
$n=500$ for a easier comparison. The results of the White test are omitted here due to its worst performance in term of size in Table \ref{table1}. As expected, for each model, the power becomes larger as the level of heteroscedasticity increases. In general, the empirical powers of all tests become smaller as the dimension $p$ goes up (ratio $p/n$ increases); the reason is that the BP test is not suitable for high-dimensional setting, and the ALRT and CVT tests are related to the degree of freedom of $k=n-p$ which becomes small when dimension $p$ increases. The CVT test is most powerful in all tested cases. 

As for the three models considered, the results for Model 1 and Model 3 are similar with each other where the BP test show no power when $p/n >0.3$ while the ALRT and CVT tests have a reasonable power unless $p/n$ is close to 1. {\colred{} Recall that in such situation, the matrix $\mathbf{X}'\mathbf{X}$ is close to singularity, the OLS estimator is performing badly. However, our procedures still show a reasonable performance.} The situation in Model 2 is radically different where the BP test has no power for all tested combinations of $(p, n)$ while the ALRT and CVT keep a reasonable power (unless $p/n$ is close to 1) as in Model 1 and 3. In conclusion, generally in all the tested situations, the proposed tests ALRT and CVT outperform the BP tests in a large extent.

\pagebreak
\begin{table}[htp]
\centering
\footnotesize
\caption{Empirical powers of the ALRT, CVT and BP tests for Model 1 under two scenarios with sample size $n=100, 500, 1000$ and varying ratio $p/n$.}
\label{table2}
\begin{tabular}{|cc|ccc|ccc|ccc|}
\hline
\multicolumn{2}{|c|}{Settings} & \multicolumn{3}{c|}{$n=100$} & \multicolumn{3}{c|}{$n=500$}  & \multicolumn{3}{c|}{$n=1000$}\\
$p_0$ & $p/n$ & ALRT & CVT & BP & ALRT & CVT &  BP & ALRT & CVT & BP\\
\hline
\multirow{6}{*}{1} & 0.05 & 0.6150 & 0.8058 & 0.9622 & 0.9970 & 1 & 1 & 1 & 1 & 1 \\
\ & 0.1 & 0.5328 & 0.7580 & 0.8258 & 0.9892 & 1 & 0.9988 &   1 & 1 & 1 \\
\ & 0.3 & 0.2782 & 0.5384 & 0.1864 & 0.8080 & 0.9872 & 0.7404 & 0.9640 & 1 & 0.9504 \\
\ & 0.5 & 0.1378 & 0.3142 & 0.0182 & 0.3858 & 0.8434 & 0.1158 & 0.6246 & 0.9810 & 0.2476 \\
\ & 0.7 & 0.0724 & 0.1274 & 0   &        0.1386 & 0.4084 & 0.0014 & 0.2060 & 0.6372 & 0.0028 \\
\ & 0.9 & 0.0566 & 0.0542 & 0   &        0.0624 & 0.0822 & 0 &  0.0596 & 0.1010 & 0 \\
\hline
\multirow{6}{*}{$0.1p$} & 0.05 &  - & - & - & 1 & 1 & 0.9964 & 1 & 1 & 0.9834 \\
\ & 0.1 & - & - & - &  1 & 1 & 0.9290 & 1 & 1 & 0.7666 \\
\ & 0.3 & 0.6732 & 0.9234 & 0.3026 & 1 & 1 & 0.2464 & 1 & 1 & 0.1356 \\
\ & 0.5 & 0.4754 & 0.8620 & 0.0418 & 1 & 1 & 0.0420 & 1 & 1 & 0.0258 \\
\ & 0.7 & 0.2024 & 0.6026 & 0.0008 & 0.9710 & 1 & 0.0032 & 1 & 1 & 0.0034 \\
\ & 0.9 & 0.0600 & 0.0916 & 0 &   0.2534 & 0.9602 & 0 &  0.4872 & 1 & 0 \\
\hline
\end{tabular}
{\em $^{*}$  \footnotesize{$``$-" denotes no suitable value} \hspace{9.5cm}}

\bigskip  Plots for the case of $n=500$ \bigskip

\includegraphics[width=0.49\textwidth]{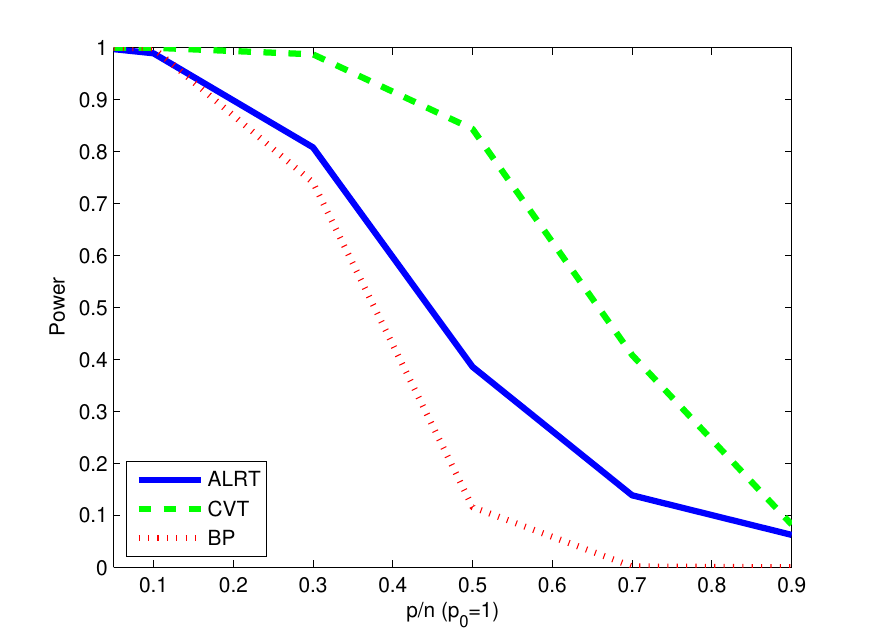}
\includegraphics[width=0.49\textwidth]{fig_T2_level1.pdf}

\end{table}


\newpage
\begin{table}[htp]
\centering
\footnotesize
\caption{Empirical powers of the ALRT, CVT and BP tests for Model 2 under two scenarios with sample size $n=100, 500, 1000$ and varying ratio $p/n$.}
\label{table3}
\begin{tabular}{|cc|ccc|ccc|ccc|}
\hline
\multicolumn{2}{|c|}{Settings} & \multicolumn{3}{c|}{$n=100$} & \multicolumn{3}{c|}{$n=500$}  & \multicolumn{3}{c|}{$n=1000$} \\
$p_0$ & $p/n$ & ALRT & CVT & BP & ALRT & CVT & BP & ALRT & CVT & BP\\
\hline
\multirow{6}{*}{1} & 0.05 & 0.9450 & 0.9276 & 0.0404 & 1 & 1 & 0.0388 & 1 & 1 & 0.0430 \\
 \ & 0.1 & 0.8436 & 0.8694 & 0.0342 & 1 & 1 & 0.0368 & 1 & 1 & 0.0338 \\
 \ & 0.3 & 0.4006 & 0.5922 & 0.0188 & 0.9240 & 0.9964 & 0.0238 & 0.9960 & 1 & 0.0208\\
 \ & 0.5 & 0.1658 & 0.3038 & 0.0060 & 0.4490 & 0.8192 & 0.0080 & 0.6826 & 0.9720 & 0.0062 \\
 \ & 0.7 & 0.0824 & 0.1138 & 0 &    0.1332 & 0.3102 & 0.0002 & 0.1846 & 0.4690 & 0.0004 \\
 \ & 0.9 & 0.0574 & 0.0504 & 0 &   0.0502 & 0.0710 & 0 & 0.0564 & 0.0732 & 0 \\
\hline
\multirow{6}{*}{$0.1p$} & 0.05 & - & - & - & 1 & 1 & 0.0484 & 1 & 1 & 0.0426 \\
 \ & 0.1 & - & - & - &                             1 & 1 & 0.0338 & 1 & 1 & 0.0376 \\
 \ & 0.3 & 0.4086 & 0.5990 & 0.0206 & 0.9262 & 0.9958 & 0.0200 & 0.9970 & 1 & 0.0232 \\
 \ & 0.5 & 0.1642 & 0.2952 & 0.0042 & 0.4460 & 0.8204 & 0.0048 & 0.6802 & 0.9714 & 0.0070 \\
 \ & 0.7 & 0.0760 & 0.1080 & 0 & 0.1376 & 0.3052 & 0.0006 & 0.1978 & 0.4642 & 0.0002 \\
 \ & 0.9 & 0.0484 & 0.0478 & 0 & 0.0556 & 0.0714 & 0 &  0.0562 & 0.0730 & 0 \\
\hline
\end{tabular}
{\em $^{*}$  \footnotesize{$``$-" denotes no suitable value}
  \hspace{9.5cm}}

\bigskip  Plots for the case of $n=500$ \bigskip

  \includegraphics[width=0.49\textwidth]{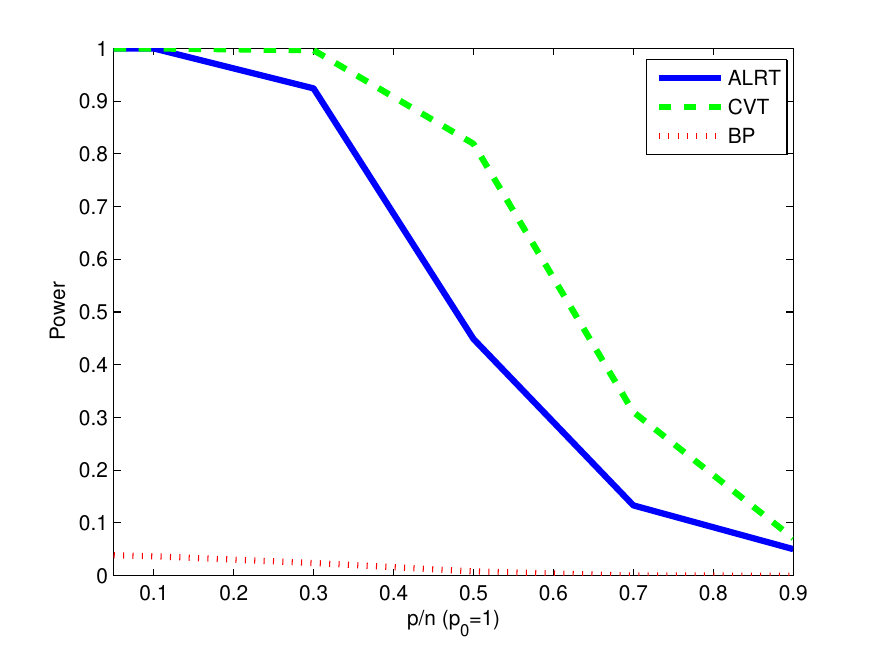}
  \includegraphics[width=0.49\textwidth]{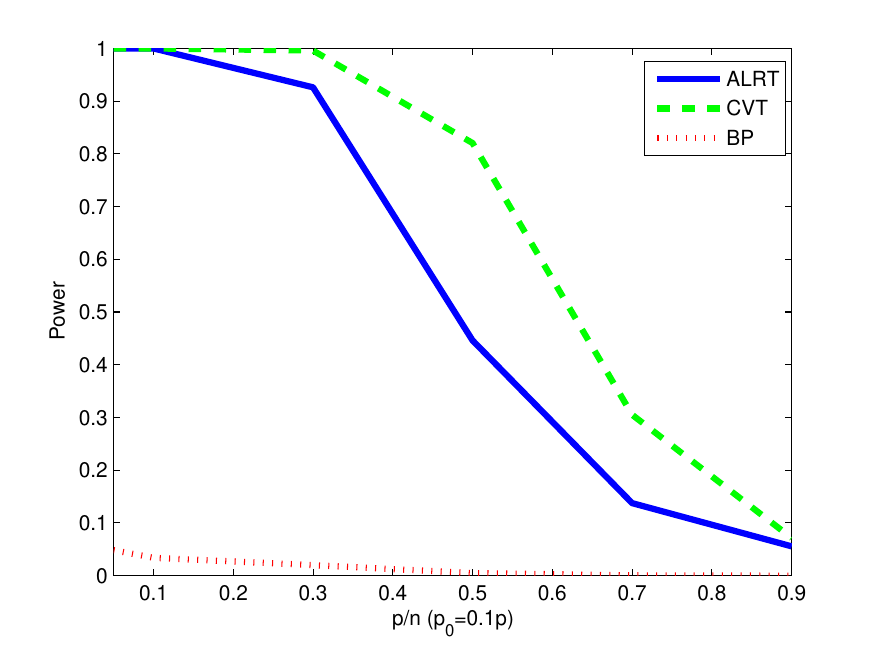}

\end{table}


\newpage 
\begin{table}[htp]
\centering
\footnotesize
\caption{Empirical powers of the ALRT, CVT and BP tests for Model 3 under two scenarios with sample size $n=100, 500, 1000$ and varying ratio $p/n$.}
\label{table4}
\begin{tabular}{|cc|ccc|ccc|ccc|}
\hline
\multicolumn{2}{|c|}{Settings} & \multicolumn{3}{c|}{$n=100$} & \multicolumn{3}{c|}{$n=500$} & \multicolumn{3}{c|}{$n=1000$} \\
$p_0$ & $p/n$ & ALRT & CVT & BP & ALRT & CVT & BP & ALRT & CVT & BP\\
\hline
\multirow{6}{*}{1} & 0.05 & 0.9648 & 0.9852 & 0.9914 & 1 & 1 & 1 & 1 & 1 & 1 \\
 \ & 0.1 & 0.9352 & 0.9706 & 0.9346 &   1 & 1 & 0.9996 & 1 & 1 & 1 \\
 \ & 0.3 & 0.5680 & 0.8346 & 0.3104 & 0.9932 & 1 & 0.8974 & 1 & 1 & 0.9886 \\
 \ & 0.5 & 0.2418 & 0.5276 & 0.0336 & 0.7298 & 0.9872 & 0.2312 & 0.9402 & 0.9998 & 0.4772 \\
 \ & 0.7 & 0.0976 & 0.2074 & 0 &  0.2336 & 0.6748 & 0.0040 & 0.3820 & 0.8960 & 0.0060 \\
 \ & 0.9 & 0.0550 & 0.0542 & 0 & 0.0638 & 0.1088 & 0 &  0.0762 & 0.1448 & 0\\
\hline
\multirow{6}{*}{$0.1p$} &0.05 & - & - & - & 1 & 1 & 0.9996 & 1 & 1 & 0.9998 \\
 \ & 0.1 & - & - & - &            1 & 1 & 0.9912 & 1 & 1 & 0.9868 \\
 \ & 0.3 & 0.7766 & 0.9578 & 0.3238 & 1 & 1 & 0.4826 & 1 & 1 & 0.4130 \\
 \ & 0.5 & 0.4034 & 0.7860 & 0.0360 & 0.9780 & 1 & 0.0682 & 1 & 1 & 0.0576 \\
 \ & 0.7 & 0.1430 & 0.3902 & 0 & 0.5108 & 0.9706 & 0.0022 & 0.7762 & 0.9996 & 0.0024 \\
 \ & 0.9 & 0.0572 & 0.0566 & 0 & 0.0942 & 0.2476 & 0 & 0.1208 & 0.3966 & 0 \\
\hline
\end{tabular}
{\em $^{*}$  \footnotesize{$``$-" denotes no suitable value} \hspace{9.5cm}}

\bigskip  Plots for the case of $n=500$ \bigskip

  \includegraphics[width=0.49\textwidth]{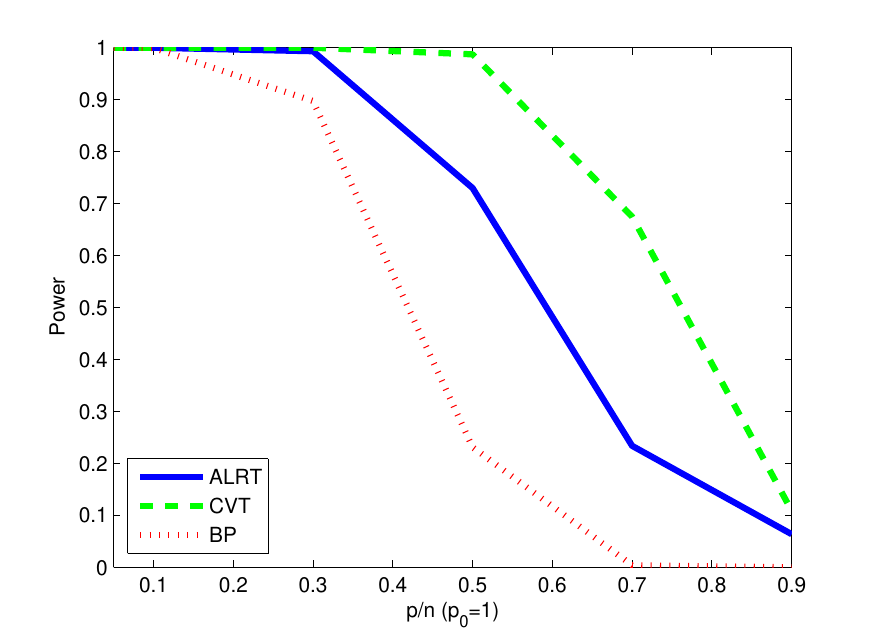}
  \includegraphics[width=0.49\textwidth]{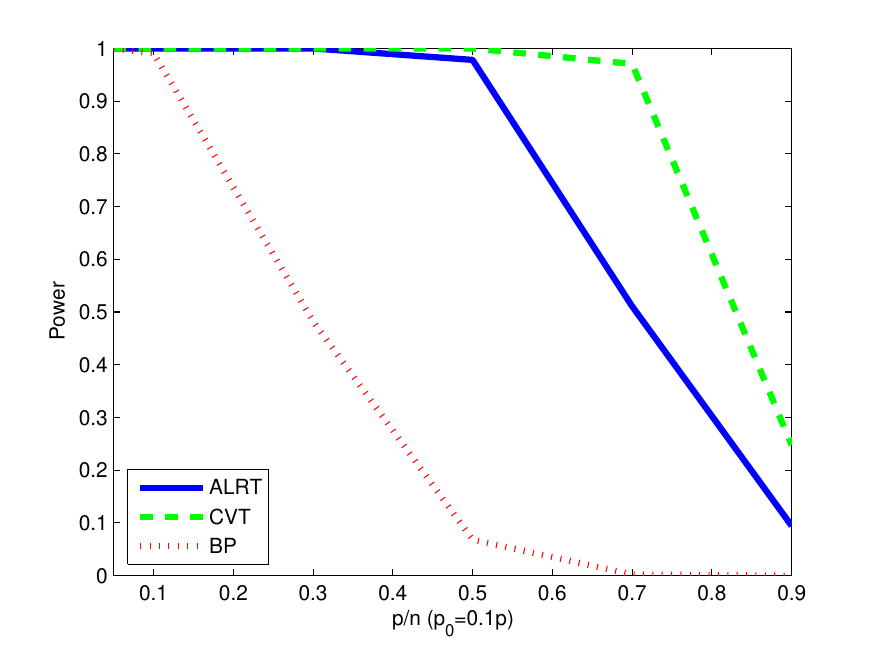}
\end{table}


\pagebreak
\subsection{Non-Gaussian design}
In this section, we investigate the performance of all tests applied to non-Gaussian design matrix. Here the entries of the design matrix $X$ are drawn from gamma distribution $G(2, 2)$ and uniform distribution $U(0, 1)$, respectively.  Except the design matrix $X$, the same setting with Section 3.1 is used to obtain the empirical sizes of all tests, and the same setting with Section 3.2 is used to obtain the empirical powers of all tests. All results are obtained using 5000 replications for each scenario.


\begin{table}
\centering
\footnotesize
\caption{Empirical sizes of the ALRT, CVT, White and BP tests for gamma and uniform designs with sample size $n=500$ and varying ratio $p/n$ (in \%).}
\label{table6}
\begin{tabular}{|c|cccc|cccc|}
\hline
\multirow{2}{*}{$p/n$} & \multicolumn{4}{c|}{Gamma design} & \multicolumn{4}{c|}{Uniform design} \\
 & ALRT & CVT & White & BP & ALRT & CVT & White & BP\\
\hline
0.05 & 4.68 & 5.80 & 0.22 & 4.12 & 4.48 & 4.84 & 0.14 & 4.18 \\
0.1 &  4.94 & 4.80 & NA & 4.28 & 5.08 & 4.84 & NA & 3.96 \\
0.3 & 5.14 & 5.62 & NA & 2.42 &  5.02 & 4.72 & NA & 2.10\\
0.5 & 5.60 & 5.76 & NA & 0.60 & 5.26 & 5.20 & NA & 0.68 \\
0.7 & 5.60 & 6.00 & NA & 0.08 & 4.86 & 4.86 & NA & 0.02 \\
0.9 & 5.72 & 6.20 & NA & 0 &  4.70 & 4.26 & NA & 0 \\
\hline
\end{tabular}
{\\ \em $^{*}$  \footnotesize{NA denotes $``$Not Applicable"} \hspace{7cm}}
\end{table}

\begin{table}
\centering
\footnotesize
\caption{Empirical powers of the ALRT, CVT and BP tests under the $p_0=0.1p$ level of heteroscedasticity for three error models with sample size $n=500$ and varying ratio $p/n$.}
\label{table7}
\begin{tabular}{|cc|ccc|ccc|}
\hline
\multicolumn{2}{|c|}{} & \multicolumn{3}{c|}{Gamma design} & \multicolumn{3}{c|}{Uniform design} \\
Setting & $p/n$ & ALRT & CVT & BP & ALRT & CVT & BP\\
\hline
\multirow{6}{*}{S1} & 0.05 & 1 & 1 & 0.9448 & 1 & 1 & 0.9998 \\
 \ & 0.1 & 1 & 1 & 0.8046 & 1 & 1 & 0.9754 \\
 \ & 0.3 & 1 & 1 & 0.4534 & 1 & 1 & 0.2964 \\
 \ & 0.5 & 1 & 1 & 0.2278 &  1 & 1 & 0.0412 \\
 \ & 0.7 & 0.9818 & 1 & 0.0304 & 0.9576 & 1 & 0.0032 \\
 \ & 0.9 & 0.2760 & 0.9978 & 0 & 0.2356 & 0.9408 & 0 \\
\hline
\multirow{6}{*}{S2} & 0.05 & 1 & 1 & 0.0366 & 1 & 1 & 0.0374 \\
 \ & 0.1 & 1 & 1 & 0.0404 & 1 & 1 & 0.0374 \\
 \ & 0.3 & 0.9222 & 0.9952 & 0.0230 & 0.9238 & 0.9942 & 0.0244 \\
 \ & 0.5 & 0.4550 & 0.8270 & 0.0056 & 0.4412 & 0.8046 & 0.0052 \\
 \ & 0.7 & 0.1410 & 0.3236 & 0.0004 & 0.1272 & 0.2914 & 0.0002 \\
 \ & 0.9 & 0.0582 & 0.0812 & 0 & 0.0542 & 0.0630 & 0 \\
\hline
\multirow{6}{*}{S3} & 0.05 & 1 & 1 & 1 & 0.6688 & 0.9138 & 0.9998 \\
 \ & 0.1 & 0.9992 & 1 & 0.9998 & 0.3910 & 0.7128 & 0.9202 \\
 \ & 0.3 & 0.3648 & 0.7522 & 0.5480 & 0.1006 & 0.1902 & 0.0804 \\
 \ & 0.5 & 0.1144 & 0.2454 & 0.0352 & 0.0712 & 0.0866 & 0.0116 \\
 \ & 0.7 & 0.0624 & 0.0946 & 0.0014 & 0.0542 & 0.0588 & 0.0002 \\
 \ & 0.9 & 0.0472 & 0.0664 & 0 &  0.0484 & 0.0456 & 0 \\
\hline
\end{tabular}
\end{table}

The empirical sizes and powers are presented in Tables~\ref{table6} and \ref{table7}, repectively. We find that there is no significant difference in terms of size and power between these two non-normal designs and the previously reported normal design. Similarly, the proposed ALRT and CVT perform well in all models and they are much better than the BP test. This suggests that the proposed tests are robust against the form or the distribution of the design matrix.

Simulation study is also conducted to explore the performance of these tests for fixed design. The design matrix $X_i$ is generated once and keep same for all replications. Even though our theoretic results are developed in the random design only, the inclusion of the fixed design simulation study is motivated by the believe that these asymptotic results of the ALRT and CVT tests remain useful in fixed design. As expected, the simulation results of empirical sizes and powers in fixed design are all similar to that in random design.  These results are omitted here for brevity. 

\subsection{Small sample sizes}

Simulation experiments are conducted to assess the performance of our tests for small sample size in a classical low-dimensional scenario. The DM test is compared here (notice that this test is not in Tables 1-4 since its implementation in a multivariate setting is unclear). Following the same set-up of \citet{dette1998testing}, the design points are chosen as $x_{i,n}=(i-1)/(n-1) (i=1, \ldots, n)$ and the sample sizes are $n=50, 25$. The BP and White tests are not considered in this part due to the fact that the design matrix in the setting considered here is nearly singular, so that the OLS estimates used by these two tests are unreliable. The considered model is $y=g(x)+0.25\sigma(x)$ with three settings:
\begin{itemize}
\item S1: $g(x)=1+\sin(x), \quad \sigma (x) = \exp (c_0 x)$,
\item S2: $g(x) =1+x, \quad \sigma (x) =  (1+c_0 \sin (10 x))^2 $,
\item S3: $g(x)=1+x, \quad \sigma (x) = (1+c_0 x)^2$,
\end{itemize}
with different values for $c_0$ (0, 0.5 and 1.0). $g(x)$ is the mean function, so the linear model tested here is one dimension. And $\sigma (x)$ is the error term. The case $c_0=0$ corresponds to the null hypothesis of homoscedasticity and the choices $c_0=0.5$ and 1 correspond to two alternatives. We calculated the proportion of rejections of the tests using 5000 simulations for each scenario. 

The empirical sizes and powers of the ALRT, CVT and DM tests are summarised in Table \ref{table5}.  The
results of the DM test are from Tables 1 and 2 of
\citet{dette1998testing}. In term of empirical size, the ALRT test is
conservative while the DM test is inclined to overestimate the size
and both of them are close to the nominal level 0.05. But the ALRT
test is more powerful than the DM test for settings S2 and S3. The
ALRT test has similar performance with the DM test in setting S1
because it runs the OLS estimation for the sinusoidal mean
function. The CVT only performs better than the DM test in term of
power in several cases. Therefore, although the ALRT test is
constructed under the high-dimensional framework, it is still a
competitive procedure in classical low-dimensional regression even
with a small sample size. This is also supported by the results for
the $p=5$ cases in Tables 1-4.

\begin{table}
\centering
\footnotesize
\caption{Empirical sizes and powers of the ALRT, CVT and DM tests under small sample size situation.}
\label{table5}
\begin{tabular}{|cc|ccc|ccc|}
\hline
\multicolumn{2}{|c|}{Settings} & \multicolumn{3}{c|}{$n=50$} & \multicolumn{3}{c|}{$n=25$} \\
Setting & $c_0$ & ALRT & CVT & DM & ALRT & CVT& DM \\
\hline
\multirow{3}{*}{S1} & 0 &0.047 & 0.032 &0.056 &0.046 & 0.020 &0.053  \\
\ & 0.5 &0.057& 0.072 &0.084 &0.051 & 0.036 &0.072 \\
\ & 1.0 &0.131 & 0.213 &0.148 &0.096 &0.093 &0.089 \\
\hline
\multirow{3}{*}{S2} & 0 &0.047 &0.032 &0.053 &0.046 &0.020 &0.052  \\
\ & 0.5 &0.601&0.485 &0.276 &0.292&0.206 &0.101 \\
\ & 1.0 &0.884 &0.792 &0.365 &0.570&0.406 &0.094 \\
\hline
\multirow{3}{*}{S3} & 0 &0.047 &0.032&0.054 &0.046&0.020 &0.053  \\
\ & 0.5 &0.094&0.135 &0.113 &0.077 & 0.061&0.076 \\
\ & 1.0 &0.250 & 0.331 &0.198 &0.152&0.145 &0.114 \\
\hline
\end{tabular}
\end{table}

\section{Real data analyses}
Though the newly proposed two tests seem to perform better than the classical ones in the simulation experiments, we now compare them on several real examples. According to the results of simulation, we use the BP test as the representation of classical tests.

\subsection{Low-dimensional data sets}
In order to check the performance of the proposed tests in low-dimensional situation, we analyse two data sets: the `bond yield' data and the `currency substitution' data\footnote{These two data sets are available in the R package `lmtest'.}. The bond yield data set is a multivariate quarterly time series from 1961(1) to 1975(4) (sample size $n=60$) with seven variables, including RAARUS (difference of interest rate on government and corporate bonds), MOOD (measure of consumer sentiment), EPI (index of employment pressure), EXP (interest rate expectations),
Y (joint proxies for the impact of callability) and K (artificial time series based on RAARUS). This data set is used to analyse the observed long-term bond yield differentials for different types of instruments. Two main works are \citet{cook1978impact} in which a linear regression of RAARUS on MOOD, EPI, EXP and RUS is fitted to find the factors contributed to the bond yield spreads, and \citet{yawitz1981measuring} in which another linear regression of RAARUS on MOOD, Y and K is fitted to see the effect of callability on bond yields. To investigate whether the homoscedasticity assumption in both models is justified, we applied the BP test, the ALRT test and the CVT test to each regression model. For the Cook-Hendershott model, we got three p-values of 0.5614 (BP), 0.3307 (ALRT) and 0.8333 (CVT). And the Yawitz-Marshall model yields three p-values of 0.3838 (BP), 0.7314 (ALRT) and 0.3885 (CVT). Hence, {\colred{} these tests show no evidence against the assumption of constant variability in both models.}

The currency substitution data set is a multivariate quarterly times series from 1960(4) to 1975(4) (sample size $n=61$) with four variables, including logCUS (logarithm of the ratio of Canadian holdings of Canadian dollar balances and Canadian holdings of U.S. dollar balances), Iu (yield on U.S. treasury bills), Ic (yield on Canadian treasury bills) and logY (logarithm of Canadian real gross national product). This data set is used to analyse the effect of flexible exchange rates and studied by \citet{bordo1982currency} where a linear model is fitted for logCUS using the other three variables as covariates. Their results were obtained under the assumption that the error variances are constant, which is supported by our proposed test: the ALRT test reports a p-value of 0.5779 and the CVT test reports a p-value of 0.1309. However, the p-value obtained by the BP test is 0.01324 which is inconsistent with the results in \citet{bordo1982currency}.

\subsection{High dimensional data sets}
In this part, we evaluate the performance of our proposed tests on two data sets with medium and high dimensions: the `international economic growth' data\footnote{Available on the web-site: https://stuff.mit.edu/~vchern/NBER/} \cite{belloni2011inference} and the `eminent-domain' data\footnote{Available on the web-site: http://faculty.chicagobooth.edu/christian.hansen/research/\#Code.} \citep{belloni2012sparse}. 

The international economic growth data set concerns the national growth rates in GDP per capita with $p=62$ covariates including education, science policies, strength of market institutions, trade openness, saving rates and others. The sample size is $n=90$. There is no unmeasurable underlying variable in this example so the regression model with all variables has constant disturbance. The CV and BP tests provide same conclusion by reporting p-values 0.5822 and 0.9436, respectively. \cite{belloni2011inference} used covariate selection procedure to select significant variables among 62 covariates and the variable $``$black market premium" is selected. Actually, this variable has important economic meaning as it characterises trade openness. Hence, the regression model without this variable will have heteroscedastic errors, and this conjecture is supported by our proposed CV test with a p-value of 0.0686 (compared with the value 0.5822 in full model). However, the p-value obtained by the BP test is 0.9808 which is inconsistent with the result in \cite{belloni2011inference}.

\citet{belloni2012sparse} studied the effects of federal appellate court decisions regarding eminent domain on a variety of economic outcomes. To explore the effect of the characteristics of three-judge panels on judicial decisions, the data set `eminent-domain' containing $p=147$ explanatory variables (gender, race, religion, political affiliation, etc.) is used with sample size $n=183$. The ratio of dimension and sample size is larger than 0.8. \citet{belloni2012sparse} argued that much heteroscedasticity exists in this data set and used heteroscedasticity consistent standard error estimator in their analysis. Applying ALRT and CVT tests on this data set, we found a p-value of $9.96\times 10^{-14}$ and $0$, respectively, strongly supporting these authors' approval. On the other hand, the BP test cannot detect the existence of heteroscedasticity by reporting a p-value of 0.3331. 

These results of real data sets analysis are consistent with the conclusion drawn from the simulation part that our newly proposed tests can provide accurate detection of heteroscedasticity under the medium or high dimensional situations, while the BP test, constructed under the low-dimensional scheme, not only cannot possess a correct size, but also loses power when heteroscedasticity exists.

\section{Conclusion and discussion}
For high-dimensional linear regression model, we propose two simple
and efficient tests to detect the existence of heteroscedasticity. The
asymptotic normalities of test statistics with simple form are
constructed under the assumption that the degree of freedom $k$ is
large compared to the sample size $n$ with $k/n\to c\in(0,1)$ as $n\to
\infty$ and are thus appropriate for analyzing  high-dimensional
data sets. Extensive Monte-Carlo experiments demonstrates the
superiority of our proposed tests over some popular existing methods
in terms of size and power.  The good performance of our tests is also
illustrated by several real data analyses. Surprisingly enough, these
high-dimensional tests  when used  in the tested low-dimensional
situations  also show a performance 
comparable to that of 
the existing classical tests which are designed specifically under 
low-dimensional scheme.  

There are still several avenues for future  research. For example, 
the
asymptotic results of the tests proposed here are based on the
normality assumption for both  the error  and the random design. 
It is highly valuable to  investigate the non-Gaussian
setting. Although we have shown some robustness of the proposed
procedures against non-Gaussian design in simulation experiments,
a thorough investigation is missing. It is however clear that new
theoretical tools will be needed to tackle with such non-Gaussian
setting. 

{\colred{}Lastly, our procedures rely on the OLS residuals, therefore have some limitations. First, it is required that $p<n$ even though both of them can be large. How to address the case where $p>n$ remains an open question. Second, it is well-known that the OLS estimates lack robustness against outliers. It is very likely that our tests possess same weakness.}

\bibliographystyle{plainnat}
\bibliography{ref_heter}

\appendix
\section{Technical proofs}

\def\bbX{\mathbf{X}} 
\def\bbV{\mathbf{V}}
\def\bbU{\mathbf{U}}
\def\bbI{\mathbf{I}}
According to (\ref{error_h1}), the OLS residuals are normal
distributed $\hat{\boldsymbol{\varepsilon}}\sim N(\mathbf{0}, \sigma^2
\mathbf{Q}_x)$, where $\mathbf{Q}_x=\bbI_n-\bbX(\bbX'\bbX)^{-1}\bbX'$ is a projection matrix of rank
$k=n-p$.
Let $\bbV=\bbX(\bbX'\bbX)^{-1/2}$. Since $\bbX$ has i.i.d. zero-mean normal
variables, it is easily seen that $\mathbf{A}\bbV$ has the same distribution as
$\bbV$ for any  $n\times n$ orthogonal matrix $\mathbf{A}$.
Therefore $\bbV$ is a
$p$-frame, that is, it is  distributed as $p$ columns of a $n\times n$
Haar matrix  \citep[Chapter 2]{muirhead82}.
Furthermore, 
since  $\mathbf{Q}_x = \bbI_n - \bbV\bbV'$, if we complement $\bbV$
to an orthogonal matrix $(\bbU,\bbV)$,  we have then $\mathbf{Q}_x=\bbU\bbU'$ and 
$\bbU$ becomes a $k$-frame ($k=n-p$) distributed as  $k$ columns of a $n\times n$
Haar matrix. 
Therefore, we have
\begin{eqnarray}\label{error_h2}
  \hat{\boldsymbol{\varepsilon}} = \mathbf{U} \mathbf{U}^\prime \boldsymbol{\varepsilon} = \mathbf{U} Z,
\end{eqnarray}
where $Z=\mathbf{U}^\prime \boldsymbol{\varepsilon}=(z_1 \ldots
z_k)^\prime \sim \mathcal{N}(\mathbf{0}, \sigma^2\mathbf{I}_k)$ under
the null hypothesis. 
Notice that despite the multiplication by $\bbU'$, $Z$ is indepedent
of $\bbU$ (since  its conditional distribution given $\bbU$ is independent
of $\bbU$). 
Rewrite $\mathbf{U}$ as
\begin{eqnarray}\label{UP}
\mathbf{U}= (\mathbf{u}_1, \ldots, \mathbf{u}_k)=\left(\begin{array}{c} \mathbf{v}_1^\prime \\ \vdots \\ \mathbf{v}_n^\prime \end{array}\right) = \left( \begin{array}{ccc} v_{11} & \cdots & v_{1k} \\ \vdots & v_{ij} & \vdots \\ v_{n1} & \cdots & v_{nk}\end{array}\right).
\end{eqnarray}
Then the components (residuals) $\{\hat{\varepsilon}_i\}_{1 \leq i \leq n}$ of $\hat{\boldsymbol{\varepsilon}}=z_1 \mathbf{u}_1 +\cdots + z_k \mathbf{u}_k$ can be expressed as
\begin{eqnarray}\label{error_h3}
\hat{\varepsilon}_i=\mathbf{v}_i^\prime Z=\sum_{j=1}^k v_{ij} z_j,\quad \textrm{for} \  i=1, \ldots, n.
\end{eqnarray}

The proofs below rely on precise properties of the $k$-frame 
$\mathbf{U}$ ($k$ columns of a Haar matrix). 
 These useful properties are recalled in the next
section, followed by the proofs of the main results of the paper. 

\subsection{Haar matrix and related results}
Here we present some important results of Haar matrix that will be used afterwards. First, the elements $\{v_{ij}\}_{1\leq j \leq k}$ of $\mathbf{v}_i$ in (\ref{UP}) have the same marginal distribution by symmetry and the square of each element has a beta distribution with parameter $\left(\frac{1}{2}, \frac{n-1}{2}\right)$, see for example \cite{reffy2005asymptotics}. Their (marginal) moments are thus easily known. For example, we have
\begin{eqnarray}\label{EMP1}
&&E\left(v_{11}^2\right) = \frac{1}{n}; \quad E\left( v_{11}^4 \right) = \frac{3}{n(n+2)}; \nonumber\\
&&E\left( v_{11}^6 \right) = \frac{15}{n(n+2)(n+4)}; \quad E\left( v_{11}^8 \right) = \frac{105}{n(n+2)(n+4)(n+6)}.
\end{eqnarray}
In addition, these elements are not independent, but weakly correlated, the moments of their products can be obtained using the following facts of an orthogonal matrix:
\begin{eqnarray*}
&&[1].\quad \sum_{i=1}^n v_{ij}^2 =1 \quad \ 1\leq j \leq k;\\
&&[2].\quad \sum_{i=1}^n v_{ij}v_{ij^\prime}=0, \quad 1\leq j\neq j^\prime \leq k.
\end{eqnarray*}
Meanwhile, by Lemma 3.4 of \cite{reffy2005asymptotics}, for positive integers $t_1, \ldots, t_s$,
\begin{eqnarray*}
E\left( v_{i_1j_1}^{t_1}\cdots v_{i_sj_s}^{t_s} \right)=0,
\end{eqnarray*}
if $\sum_{i_\alpha=u} t_\alpha$ is odd for some $1\leq u\leq n$, or $\sum_{j_\alpha=w} t_\alpha$ is odd for some $1\leq w \leq n$.
This leads to the following list of cross-moment identities that will be used in upcoming proofs. The cross-moments of two elements in a same row (or same column) are as follows
\begin{eqnarray}\label{EMP2}
&& E \left(v_{11}^2 v_{12}^2\right) = \frac{1}{n(n+1)};\nonumber\\
&& E\left(v_{11}^4 v_{12}^2\right) = \frac{3}{n(n+2)(n+4)};\nonumber \\
&& E\left( v_{11}^6 v_{12}^2\right) = \frac{15}{n(n+2)(n+4)(n+6)};\nonumber\\
&& E\left( v_{11}^4 v_{12}^4\right) = \frac{9n-6}{n(n+2)^2(n+4)(n+6)}.
\end{eqnarray}
The cross-moments of two elements in different rows and different columns are
\begin{eqnarray}\label{EMP3}
&&E\left( v_{11}^4 v_{22}^2\right) = \frac{3(n+3)}{n(n-1)(n+2)(n+4)};\nonumber\\
&&E\left( v_{11}^4 v_{22}^4\right)  = \frac{9 n^2 +81n +222}{n(n-1)(n+2)^2(n+4)(n+6)}.
\end{eqnarray}
The cross-moments of three elements in a same row (or same column) are
\begin{eqnarray}\label{EMP4}
&& E\left(v_{11}^2 v_{12}^2 v_{13}^2\right) = \frac{1}{n(n+2)(n+4)};\nonumber\\
&& E\left(v_{11}^4 v_{12}^2 v_{13}^2\right) = \frac{3(n^2+4)}{n(n-2)(n+2)^2(n+4)(n+6)}.
\end{eqnarray}
The cross-moments of three elements in different rows or different columns are
\begin{eqnarray}\label{EMP5}
&&E\left( v_{11}^2 v_{12}^2 v_{22}^2\right) = \frac{n+1}{n(n-1)(n+2)(n+4)};\nonumber\\[1mm]
&&E\left(v_{11}^2 v_{12}^2 v_{23}^2\right) = \frac{n+3}{n(n-1)(n+2)(n+4)};\nonumber\\[1mm]
&&E\left( v_{11}^4 v_{21}^2 v_{22}^2\right) = \frac{3 n^2+15n+42}{n(n-1)(n+2)^2(n+4)(n+6)};\nonumber\\[1mm]
&&E\left( v_{11}^4 v_{22}^2 v_{23}^2\right) = \frac{3n^3 +21n^2 +12n-156}{n(n-1)(n-2)(n+2)^2(n+4)(n+6)}.
\end{eqnarray}
The cross-moments of four elements in the same row (or same column) is
\begin{eqnarray}\label{EMP6}
 E\left( v_{11}^2 v_{12}^2 v_{13}^2 v_{14}^2 \right) = \frac{n^3-3n^2-4n-60}{n(n-2)(n-3)(n+2)^2(n+4)(n+6)}.
\end{eqnarray}
The cross-moments of four elements in different rows or different columns are
\begin{eqnarray}\label{EMP7}
&&E\left( v_{11}^2 v_{12}^2 v_{21}^2 v_{22}^2\right) = \frac{n^3 +3n^2 -4n-36}{n(n-1)(n-2)(n+2)^2(n+4)(n+6)};\nonumber\\[1mm]
&&E\left( v_{11}^2 v_{12}^2 v_{21}^2 v_{23}^2\right) = \frac{n^4 +3n^3 -10n^2 -36n+96}{n(n-1)(n^2-4)^2(n+4)(n+6)};\nonumber\\[1mm]
&&E\left( v_{11}^3 v_{12} v_{21} v_{22}\right) = -\frac{3}{n(n-1)(n+2)(n+4)};\nonumber\\[1mm]
&&E\left( v_{11}^2 v_{12}^2 v_{23}^2 v_{24}^2\right) = \frac{n^4 +5n^3 -10n^2 -44n+120}{n(n-1)(n^2-4)^2(n+4)(n+6)};\nonumber\\[1mm]
&&E\left( v_{11}^3 v_{12} v_{21}^3 v_{22}\right) = -\frac{9n-6}{n(n-1)(n+2)^2(n+4)(n+6)};\nonumber\\[1mm]
&& E\left( v_{11}^3 v_{12} v_{21} v_{22}^3\right) \approx E\left( v_{11}^3 v_{12} v_{21}^3 v_{22}\right).
\end{eqnarray}
The last approximate expression is due to the symmetry between the elements. Finally, some useful cross-moments of more than four elements in different rows or different columns are as follows:
\begin{eqnarray}\label{EMP8}
&&E\left( v_{11}^2 v_{12} v_{13} v_{22} v_{23}\right) = -\frac{1}{n(n-1)(n+2)(n+4)};\nonumber\\[1mm]
&&E\left(v_{11}^3 v_{12}v_{21} v_{22} v_{23}^2\right) = -\frac{3n^2 -6n -48}{n(n-1)(n-2)(n+2)^2(n+4)(n+6)};\nonumber\\[1mm]
&&E\left( v_{11}^2 v_{12} v_{13} v_{21}^2 v_{22} v_{23}\right) = -\frac{n^3 -6n^2 +20n -48}{n(n-1)(n^2-4)^2(n+4)(n+6)};\nonumber\\[1mm]
&&E\left( v_{11} v_{12} v_{13} v_{14} v_{21} v_{22} v_{23} v_{24}\right) = \frac{3(n^3 -6n^2 +20n-48)}{n(n-1)(n-3)(n^2-4)^2(n+4)(n+6)};\nonumber\\[1mm]
&&E\left( v_{11}^2 v_{12} v_{13} v_{22} v_{23} v_{24}^2\right)\approx E\left( v_{11}^2 v_{12} v_{13} v_{21}^2 v_{22} v_{23}\right).
\end{eqnarray}

Next, by Theorem 2.1 of \cite{song1997}, the joint distribution of all the squared elements in $\mathbf{v}_i $ in (\ref{UP}) $(1\leq i \leq n)$ is known to be
\begin{eqnarray}
\left( v_{i1}^2, v_{i2}^2, \ldots, v_{ik}^2 \right) \sim D_k \left( \frac{1}{2}, \cdots, \frac{1}{2}; \frac{n-k}{2} \right),
\end{eqnarray}
where $D_k(\alpha_1,\ldots, \alpha_k; \alpha_{k+1})$ is the Dirichlet distribution with positive parameters $(\alpha_1,\ldots, \alpha_k; \alpha_{k+1})$. Therefore, $||\mathbf{v}_i||^2 = v_{i1}^2 +\cdots + v_{ik}^2$ has beta distribution with parameters $\left( \frac{k}{2}, \frac{n-k}{2}\right)$. It follows that
\begin{eqnarray}\label{lemma3con1}
E\left( ||\mathbf{v}_i||^2 \right)=c_n, \quad var\left( ||\mathbf{v}_i||^2 \right)=\frac{2c_n(1-c_n)}{n+2},
\end{eqnarray}
\begin{eqnarray}\label{lemma3con2}
cov\left( ||\mathbf{v}_i||^2,  ||\mathbf{v}_j||^2\right) =\frac{2c_n(c_n-1)}{(n-1)(n+2)}, \quad \textrm{for} \ i\neq j,
\end{eqnarray}
\begin{eqnarray}\label{logterm}
E\left(\log||\mathbf{v}_i||^2 \right) = \log c_n +\frac{1}{n} -\frac{1}{k}+O\left(\frac{1}{n^2}\right)+O\left(\frac{1}{k^2}\right),
\end{eqnarray}
\begin{eqnarray}\label{onelogterm}
E\left(||\mathbf{v}_i||^2 \log||\mathbf{v}_i||^2 \right)
= c_n \left( \log c_n +\frac{1}{k} -\frac{1}{n} \right) +O\left(\frac{1}{n^2}\right)+O\left(\frac{1}{k^2}\right),
\end{eqnarray}
\begin{eqnarray}\label{logtermsquare}
E \left(\log ||\mathbf{v}_i||^2\right)^2
&=& \left( \log c_n \right)^2+2 \left(\frac{1}{n} -\frac{1}{k} \right)  \log c_n + \frac{2}{k} -\frac{2}{n}\nonumber\\[1mm]
 && \ +O\left(\frac{1}{n^2}\right) +O\left(\frac{1}{k^2}\right)+ O\left(\frac{1}{nk}\right),
\end{eqnarray}
Next, we derive the asymptotic limits for some joint distributions of $\left\{||\mathbf{v}_i||^2, \log ||\mathbf{v}_i||^2\right\}$.
\begin{lemma}\label{lemma3}
Based on the above results on $||\mathbf{v}_i||^2, 1\leq i \leq n$, as $k, n\to \infty$, we have
\begin{eqnarray}\label{lemma3result}
\sqrt{\frac{n}{2c_n(1-c_n)}} \left( \begin{array}{c} ||\mathbf{v}_1||^2 -c_n \\ ||\mathbf{v}_2||^2 -c_n \end{array}\right) \overset{\mathcal{D}}{\longrightarrow} \mathcal{N}\left(\mathbf{0}, \mathbf{I}_2 \right).
\end{eqnarray}
\end{lemma}
\begin{proof}
For $||\mathbf{v}_i||^2, 1\leq i \leq n$, the multivariate central limit theorem states that
\begin{eqnarray*}
\mathbf{\Sigma}_0^{-1/2} \cdot \sqrt{n}\left( \begin{array}{c} ||\mathbf{v}_1||^2 -c_n \\ ||\mathbf{v}_2||^2 -c_n \end{array}\right) \overset{\mathcal{D}}{\longrightarrow}\mathcal{N}\left(\mathbf{0}, \mathbf{I}_2 \right),
\end{eqnarray*}
where
\begin{eqnarray*}
\mathbf{\Sigma}_0=\left( \begin{array}{cc} n\cdot var(||\mathbf{v}_1||^2) & n \cdot cov(||\mathbf{v}_1||^2, ||\mathbf{v}_2||^2) \\ n \cdot cov(||\mathbf{v}_1||^2, ||\mathbf{v}_2||^2) & n\cdot var(||\mathbf{v}_2||^2)\end{array}\right).
\end{eqnarray*}
By the previous results (\ref{lemma3con1}) and (\ref{lemma3con2}), we obtain that
\begin{eqnarray*}
&&n\cdot var(||\mathbf{v}_1||^2)= n\cdot var(||\mathbf{v}_2||^2)=2c_n(1-c_n),\\
&&n \cdot cov(||\mathbf{v}_1||^2, ||\mathbf{v}_2||^2)= \frac{2c_n(c_n-1)}{n} \to 0 \quad \textrm{as} \ n\to \infty.
\end{eqnarray*}
Then, Lemma \ref{lemma3} follows.
\end{proof}

There are two corollaries (easy consequences) of (\ref{lemma3result}) by delta method:
\begin{eqnarray*}
\sqrt{n} \left( \begin{array}{c} \left(\sqrt{2c_n(1-c_n)}\right)^{-1}\left(||\mathbf{v}_1||^2 -c_n\right) \\ \left(\sqrt{2(1-c_n)/c_n}\right)^{-1} \left(\log \left(||\mathbf{v}_2||^2\right) -\log c_n\right)\end{array}\right) \overset{\mathcal{D}}{\longrightarrow} \mathcal{N}\left(\mathbf{0}, \mathbf{I}_2  \right),
\end{eqnarray*}
and
\begin{eqnarray*}
\sqrt{nc_n/2(1-c_n)} \left(\begin{array}{c} \log \left(||\mathbf{v}_1||^2\right) -\log c_n \\ \log \left(||\mathbf{v}_2||^2\right) -\log c_n \end{array}\right) \overset{\mathcal{D}}{\longrightarrow} \mathcal{N}\left(\mathbf{0}, \mathbf{I}_2\right).
\end{eqnarray*}
Then, by these two corollaries, we obtain the following useful results
\begin{eqnarray}\label{twolog}
E \left(\log||\mathbf{v}_1||^2 \log||\mathbf{v}_2||^2 \right) 
&=& \left( \log c_n \right)^2+2 \left(\frac{1}{n} -\frac{1}{k} \right)  \log c_n +O\left(\frac{1}{n^2}\right) \nonumber\\
&& \quad +O\left(\frac{1}{k^2}\right) +O\left(\frac{1}{nk}\right),
\end{eqnarray}
\begin{eqnarray}\label{onelog}
E \left( ||\mathbf{v}_1||^2 \log||\mathbf{v}_2||^2 \right) = c_n\left(\log c_n+\frac{1}{n}-\frac{1}{k}\right)+O \left(\frac{1}{n^2}\right) +O\left(\frac{1}{k^2}\right).
\end{eqnarray}
Notice that the crucial condition $\liminf c_n >0$ in Assumption (d) is here used to ensure the well-definiteness of the centering term $\log c_n$.

\subsection{Proof of Lemma \ref{joint1}}
Recall that $\hat{\boldsymbol{\varepsilon}}=\mathbf{U}\mathbf{Z}$ is
distributed as a degenerated $p$-dimensional Gaussian vector of rank
$k=n-p$. Therefore, by standard central limit theory  $\left( \sum_{i=1}^n
  \hat{\varepsilon}_i^2, \sum_{i=1}^n \log \hat{\varepsilon}_i^2
\right)$ is asymptotically Gaussian after suitable centering and
normalization when $k\to \infty$. It remains to determine their
limiting mean and variance-covariances. 
 
\emph{\bf Moments of $\sum_{i=1}^n \hat{\varepsilon}_i^2$.}
According to (\ref{error_h2}) $\hat{\boldsymbol{\varepsilon}}=\mathbf{U}\mathbf{Z}$, then
\begin{eqnarray}
\sum_{i=1}^n \hat{\varepsilon}_i^2=\hat{\boldsymbol{\varepsilon}}^\prime \hat{\boldsymbol{\varepsilon}}=\mathbf{Z}^\prime \mathbf{U}^\prime  \mathbf{U}\mathbf{Z}=\mathbf{Z}^\prime \mathbf{Z}=\chi_k^2,
\end{eqnarray}
is a chi-square distributed random variable with degree of freedom $k$ due to $\mathbf{U}^\prime  \mathbf{U}=\mathbf{I}_k$. Therefore, the expectation and variance of $\sum_{i=1}^n \hat{\varepsilon}_i^2$ are
\begin{eqnarray}
E\left(\sum_{i=1}^n \hat{\varepsilon}_i^2\right)=k, \quad var\left(\sum_{i=1}^n \hat{\varepsilon}_i^2\right)=2k.
\end{eqnarray}

\emph{\bf Moments of $\sum_{i=1}^n \log \hat{\varepsilon}_i^2$.}
By equation (\ref{error_h3}), when given the vector $\mathbf{v}_1$, $\hat{\varepsilon}_i$ is normally distributed with zeros mean and the variance is $||\mathbf{v}_1||$, which is the $L_2$-norm of $\mathbf{v}_1$. Denote that $\hat{\varepsilon}_i=||\mathbf{v}_i||\eta_i$, where $\eta_i$ is standard normal distributed.

The expectation of $\sum_{i=1}^n \log \hat{\varepsilon}_i^2$ is calculated as follows
\begin{eqnarray}
E\left( \sum_{i=1}^n \log \hat{\varepsilon}_i^2 \right) 
= n E \left[ E \left( \log (||\mathbf{v}_1||^2 \eta_1^2)\big|\mathbf{v}_1  \right) \right]
= n E\left[ -\gamma -\log 2 +\log(||\mathbf{v}_1||^2) \right],
\end{eqnarray}
and by the previous result (\ref{logterm}), we obtain
\begin{eqnarray}
M_2=E\left( \sum_{i=1}^n \log \hat{\varepsilon}_i^2 \right) = n\left( \log c_n - \gamma - \log 2 +\frac{1}{n} -\frac{1}{k} +O\left(\frac{1}{n^2}\right) +O\left(\frac{1}{k^2}\right) \right).
\end{eqnarray}

The variance of $\sum_{i=1}^n \log \hat{\varepsilon}_i^2$ is calculated as follows
\begin{eqnarray}\label{var_log}
var \left( \sum_{i=1}^n \log \hat{\varepsilon}_i^2 \right) &=& E \left( \sum_{i=1}^n \log \hat{\varepsilon}_i^2 \right)^2 - E^2\left( \sum_{i=1}^n \log \hat{\varepsilon}_i^2 \right) \nonumber \\[1mm]
&=& n E\left( \log \hat{\varepsilon}_1^2 \right)^2 + n(n-1) E \left(  \log \hat{\varepsilon}_1^2\cdot \log \hat{\varepsilon}_2^2 \right) - M_2^2,
\end{eqnarray}
where $E\left( \log \hat{\varepsilon}_1^2 \right)^2$ is obtained by the previous results (\ref{logterm}) and (\ref{logtermsquare})
\begin{eqnarray}\label{var_log1}
&&E\left( \log \hat{\varepsilon}_1^2 \right)^2\nonumber \\
&&= \frac{\pi^2}{2}+(\log 2)^2 +\left(\log c_n\right)^2  +2\gamma \log 2 -2(\gamma +\log 2)\left(\log c_n +\frac{1}{n}-\frac{1}{k}\right) \nonumber \\
&& \quad +\gamma^2+ 2\left(\frac{1}{n}-\frac{1}{k}\right)\log c_n +\frac{2}{k}-\frac{2}{n}+O\left(\frac{1}{n^2}\right) +O\left(\frac{1}{k^2}\right) +O\left(\frac{1}{nk}\right),
\end{eqnarray}
and
\begin{eqnarray*}
E \left(  \log \hat{\varepsilon}_1^2\cdot \log \hat{\varepsilon}_2^2 \right)
&=& E \Big\{ \log||\mathbf{v}_1||^2 \log||\mathbf{v}_2||^2 + (\gamma +\log 2)^2 \nonumber \\
&& \quad \quad -(\gamma +\log 2) \left(\log||\mathbf{v}_1||^2 + \log||\mathbf{v}_2||^2 \right) \Big\},
\end{eqnarray*}
and by the previous results (\ref{logterm}) and (\ref{twolog}), we obtain
\begin{eqnarray}\label{var_log2}
&&E \left(  \log \hat{\varepsilon}_1^2\cdot \log \hat{\varepsilon}_2^2 \right) \nonumber \\
&&= (\log c_n)^2 + 2 \left(\frac{1}{n}-\frac{1}{k}\right)\log c_n -2(\gamma +\log 2) \left( \log c_n +\frac{1}{n} -\frac{1}{k}\right) \nonumber \\
&&\quad  +(\gamma +\log 2)^2  +O\left(\frac{1}{n^2}\right) + O\left(\frac{1}{k^2}\right) +O\left(\frac{1}{nk}\right).
\end{eqnarray}
Then, we get the variance by substituting (\ref{var_log1}) and (\ref{var_log2}) in (\ref{var_log})
\begin{eqnarray}
var \left( \sum_{i=1}^n \log \hat{\varepsilon}_i^2 \right)=n \left( \frac{\pi^2}{2} +\frac{2}{c_n}-2 +O\left(\frac{1}{n}\right) +O\left(\frac{1}{k}\right)\right).
 \end{eqnarray}

The covariance of $\sum_{i=1}^n \hat{\varepsilon}_i^2$ and $\sum_{i=1}^n \log\hat{\varepsilon}_i^2$ is
\begin{eqnarray}\label{cov_log}
cov\left(\sum_{i=1}^n \hat{\varepsilon}_i^2, \sum_{i=1}^n \log\hat{\varepsilon}_i^2\right) = n E \left( \hat{\varepsilon}_1^2 \log \hat{\varepsilon}_1^2 \right) + n(n-1) E \left(\hat{\varepsilon}_1^2 \log \hat{\varepsilon}_2^2 \right) - kM_2.
\end{eqnarray}
By the previous results (\ref{lemma3con1}) and (\ref{onelogterm}), we obtain
\begin{eqnarray}\label{cov_log1}
E \left( \hat{\varepsilon}_1^2 \log \hat{\varepsilon}_1^2 \right) = c_n \left( \log c_n +2-\gamma -\log2 +\frac{1}{n}-\frac{1}{k} \right) +O\left(\frac{1}{n^2}\right) +O\left(\frac{1}{k^2}\right),
\end{eqnarray}
and by the previous results (\ref{lemma3con1}) and (\ref{onelog}), we have
\begin{eqnarray}\label{cov_log2}
E \left(\hat{\varepsilon}_1^2 \log \hat{\varepsilon}_2^2 \right) = c_n\left(\log c_n-\gamma \log 2 +\frac{1}{n} -\frac{1}{k}\right)+O \left(\frac{1}{n^2}\right) +O \left(\frac{1}{k^2}\right).
\end{eqnarray}
Then, we get the covariance by substituting (\ref{cov_log1}) and (\ref{cov_log2}) in (\ref{cov_log})
\begin{eqnarray}
cov\left(\sum_{i=1}^n \hat{\varepsilon}_i^2, \sum_{i=1}^n \log\hat{\varepsilon}_i^2\right) =n \left(2 +O(1/n)+O(1/k)\right).
\end{eqnarray}
The proof of Lemma 1 is complete.

\subsection{Proof of Theorem \ref{T1}}
Define two sequences $X_n$ and $Y_n$ as
\begin{eqnarray*}
\left(\begin{array}{c} X_n \\ Y_n \end{array}\right) = n^{-1/2} \left( \begin{array}{l} \sum_{i=1}^n \hat{\varepsilon}_i^2 -n\cdot c_n\\ \sum_{i=1}^n \log \hat{\varepsilon}_i^2 -n\cdot  \left(\log c_n-\gamma -\log2\right) \end{array}\right).
\end{eqnarray*}
The result of Lemma \ref{joint1} can be rewritten as
\begin{eqnarray*}
\left(\frac{1}{n}\mathbf{\Sigma}_1\right)^{-1/2} \left( \begin{array}{c} X_n \\ Y_n \end{array} \right) \overset{\mathcal{D}}{\longrightarrow} \mathcal{N}\left(\mathbf{0}, \mathbf{I}_2\right).
\end{eqnarray*}
Let $a=c_n$, $b=\left(\log c_n-\gamma -\log2\right)$.
By definition, we have
\begin{eqnarray*}
\frac{1}{n}\sum_{i=1}^n \hat{\varepsilon}_i^2 =a +\frac{1}{\sqrt{n}}X_n, \quad \frac{1}{n}\sum_{i=1}^n \log \hat{\varepsilon}_i^2 = b +\frac{1}{\sqrt{n}}Y_n. 
\end{eqnarray*}
Then, the statistic $T_1$ can be rewritten as
\begin{eqnarray*}
T_1=a \exp\left(-b\right) \left[ 1+\frac{1}{a\sqrt{n}}X_n -\frac{1}{\sqrt{n}}Y_n +O_p\left(\frac{1}{n}\right) \right].
\end{eqnarray*}
And
\begin{eqnarray*}
\sqrt{n}\cdot T_1= \sqrt{n}\log (a\exp(-b)) +\frac{1}{a}X_n-Y_n +O_p\left(\frac{1}{\sqrt{n}}\right).
\end{eqnarray*}
Therefore, $\sqrt{n} T_1$ is asymptotic Gaussian, and its limiting parameters are
\begin{eqnarray*}
E\left( \sqrt{n} T_1 \right)= \sqrt{n} (\gamma +\log 2) + o(\sqrt{n}), \textrm{and} \ var\left( \sqrt{n} T_1 \right)=\frac{\pi^2}{2} -2 +o(1).
\end{eqnarray*}
The proof of Theorem \ref{T1} is complete.

\subsection{Proof of Lemma \ref{joint2}}
Recall that $\hat{\boldsymbol{\varepsilon}}=\mathbf{U}\mathbf{Z}$ is
distributed as a degenerated $p$-dimensional Gaussian vector of rank
$k=n-p$. By standard central limit theory  $\left( \sum_{i=1}^n \hat{\varepsilon}_i^4,
  \sum_{i=1}^n \hat{\varepsilon}_i^2 \right)$ is asymptotic Gaussian
up to suitable centering and normalization when $k\to \infty$. It
remains to determine its limiting mean and variance-covariances. 

According to (\ref{error_h3}) $\hat{\varepsilon}_i=\sum_{j=1}^k v_{ij} z_j, 1\leq i \leq n$, the expectation and variance of $\sum_{i=1}^n \hat{\varepsilon}_i^4$ are expanded in terms of $\{v_{ij}\}$ and $\{z_j\}$.
First, the expectation of $\sum_{i=1}^n \hat{\varepsilon}_i^4$ is calculated as
\begin{eqnarray*}
E\left(\sum_{i=1}^n \hat{\varepsilon}_i^4\right)
&=& n E\left( \sum_{j_1, j_2, j_3, j_4=1}^k v_{1j_1} v_{1j_2} v_{1j_3} v_{1j_4} z_{j_1} z_{j_2} z_{j_3} z_{j_4}\right)\nonumber\\[1mm]
&=& n \left[ 3kE \left(v_{11}^4\right)+k(k-1) E\left(v_{11}^2 v_{12}^2\right) \right],
\end{eqnarray*}
and by the moment identities (\ref{EMP1}) and (\ref{EMP2}), we obtain
\begin{eqnarray}
M_1
= E\left(\sum_{i=1}^n \hat{\varepsilon}_i^4\right)
= \frac{3k(k+2)}{n+2}.
\end{eqnarray}
Second, the variance of $\sum_{i=1}^n \hat{\varepsilon}_i^4$ is calculated as
\begin{eqnarray}\label{var_4}
var\left(\sum_{i=1}^n \hat{\varepsilon}_i^4\right)
&=& E \left(\sum_{i=1}^n \hat{\varepsilon}_i^4\right)^2
- E^2 \left(\sum_{i=1}^n \hat{\varepsilon}_i^4\right)\nonumber\\[2mm]
&=& n E \left(\hat{\varepsilon}_1^8\right) + n (n-1) E\left(\hat{\varepsilon}_1^4 \hat{\varepsilon}_2^4\right) -M_1^2,
\end{eqnarray}
where
\begin{eqnarray*}
E \left(\hat{\varepsilon}_1^8\right)
&=& 105 k E\left( v_{11}^8 \right) + 420 k (k-1) E\left( v_{11}^6 v_{12}^2\right)\nonumber\\
&& +315 k (k-1) E\left( v_{11}^4 v_{12}^4\right)
+630 k(k-1)(k-2) E\left(v_{11}^4 v_{12}^2 v_{13}^2\right)\nonumber \\
&&+ 105 k(k-1)(k-2)(k-3) E\left( v_{11}^2 v_{12}^2 v_{13}^2 v_{14}^2 \right),
\end{eqnarray*}
and by the moment identities (\ref{EMP1}), (\ref{EMP2}), (\ref{EMP4}) and (\ref{EMP6}), we obtain
\begin{eqnarray}\label{8th}
E \left(\hat{\varepsilon}_1^8\right)&=&\Big[ 105 k^4 \left( n^3 -3n^2-4n-60 \right) + 1260 k^3 \left( n^3-3n^2+8n+12 \right) \nonumber \\
&&+ 2520 k \left( 2n^3-3n^2 -11n +24 \right) +420 k^2 \left( 11n^3-51n^2-62n+150 \right) \Big]\nonumber \\
&&\times \left[n(n-2)(n-3)(n+2)^2(n+4)(n+6)\right]^{-1};
\end{eqnarray}
and by the moment identities (\ref{EMP2}), (\ref{EMP3}), (\ref{EMP5}), (\ref{EMP7}) and (\ref{EMP8}), we obtain
\begin{eqnarray}\label{inter44}
E\left(\hat{\varepsilon}_1^4 \hat{\varepsilon}_2^4\right)&=& \Big[ 9 k^4 \left( n^4 +5n^3 -10n^2-44n +120 \right) + 108 k^3 \left( n^4 +3n^3 -10 n^2  \right) \nonumber\\
&&\quad +108k^3 (-44n+88)+ 36 k^2 \left( 11n^4 -5n^3 +16n^2 -334n +384 \right) \nonumber\\
&&\quad +72 k \left( 6n^4 -77n^3 +157n^2 +116n -304 \right) \Big]\nonumber\\
&& \times \left[ n(n-1)(n^2-4)^2(n+4)(n+6) \right]^{-1}.
\end{eqnarray}
Then, by substituting equations (\ref{8th}) and (\ref{inter44}) into (\ref{var_4}), the variance of $\sum_{i=1}^n \hat{\varepsilon}_i^4$ is
\begin{eqnarray}
var\left(\sum_{i=1}^n \hat{\varepsilon}_i^4\right)&=& \Big[ 24 k^4 \left( n^4 +10n^3 -121n^2+152n -78\right) \nonumber\\
&& \quad +72 k^3 \left( n^5 +7n^4 -50n^3 +384n^2 -1132n +1200\right) \nonumber\\
&&  \quad +24 k^2 \left( 15n^5 +89n^4 -751n^3 -3245n^2 +18394n-20514\right) \nonumber\\
&& \quad +72 k \left( 6n^5 -n^4 -63n^3+498n^2 -1882n+2208\right) \Big]\nonumber\\
&& \times \left[(n-3)(n^2-4)^2(n+4)(n+6)\right]^{-1}.
\end{eqnarray}
Lastly, the covariance of $\sum_{i=1}^n \hat{\varepsilon}_i^2$ and $\sum_{i=1}^n \hat{\varepsilon}_i^4$ is calculated as follows
\begin{eqnarray}\label{cov2}
cov\left(\sum_{i=1}^n \hat{\varepsilon}_i^2, \sum_{i=1}^n \hat{\varepsilon}_i^4\right) &=& E\left(\sum_{i=1}^n \hat{\varepsilon}_i^2\cdot \sum_{i=1}^n \hat{\varepsilon}_i^4\right) -E\left(\sum_{i=1}^n \hat{\varepsilon}_i^2\right) E\left(\sum_{i=1}^n \hat{\varepsilon}_i^4\right)\nonumber\\[2mm]
&=& n E\left( \hat{\varepsilon}_1^6\right) +n (n-1) E\left( \hat{\varepsilon}_1^4 \hat{\varepsilon}_2^2 \right)-kM_1,
\end{eqnarray}
where
\begin{eqnarray*}
E\left( \hat{\varepsilon}_1^6\right) &=& k E\left( v_{11}^6\right) E\left(z_1^6\right) + 15 k(k-1) E\left(v_{11}^4 v_{12}^2\right) E\left(z_1^4 z_2^2\right)\nonumber\\
 &&+ 15 k(k-1)(k-2) E\left(v_{11}^2 v_{12}^2 v_{13}^2\right) E\left( z_1^2 z_2^2 z_3^3\right),
\end{eqnarray*}
and by the moment identities (\ref{EMP1}), (\ref{EMP2}) and (\ref{EMP4}), we obtain
\begin{eqnarray}\label{6th}
E\left( \hat{\varepsilon}_1^6\right)= \frac{15 k^3 +90k^2+110k}{n(n+2)(n+4)},
\end{eqnarray}
and
\begin{eqnarray*}
E\left( \hat{\varepsilon}_1^4 \hat{\varepsilon}_2^2 \right) &=& k E\left(v_{11}^4 v_{21}^2\right) E\left(z_1^6\right) + k(k-1) E\left( v_{11}^4 v_{22}^2\right) E\left( z_1^4 z_2^2\right) \nonumber\\
&& + 6k(k-1) E\left( v_{11}^2 v_{12}^2 v_{22}^2\right) E\left(z_1^4 z_2^2\right)
 + 8k(k-1) E\left( v_{11}^3 v_{12} v_{21} v_{22}\right) E\left(z_1^4 z_2^2\right)\nonumber \\
&& + 3k (k-1)(k-2) E\left(v_{11}^2 v_{12}^2 v_{23}^2\right) E\left( z_1^2 z_2^2 z_3^2\right) \nonumber\\
&& + 12 k(k-1)(k-2) E\left( v_{11}^2 v_{12} v_{13} v_{22} v_{23}\right) E\left( z_1^2 z_2^2 z_3^2\right),
\end{eqnarray*}
and by the moment identities (\ref{EMP3}), (\ref{EMP5}), (\ref{EMP6}) and (\ref{EMP8}), we obtain
\begin{eqnarray}\label{inter42}
E\left( \hat{\varepsilon}_1^4 \hat{\varepsilon}_2^2 \right)=\frac{3n k^3 -3k^3 +18nk^2 -18k^2}{n(n-1)(n+2)(n+4)}.
\end{eqnarray}
Then, we get the covariance by substituting (\ref{6th}) and (\ref{inter42}) into (\ref{cov2})
\begin{eqnarray}
cov\left(\sum_{i=1}^n \hat{\varepsilon}_i^2, \sum_{i=1}^n \hat{\varepsilon}_i^4\right) = \frac{12 k^2 (n+4) +110k}{(n+2)(n+4)}.
\end{eqnarray}
The proof of Lemma \ref{joint2} is complete.

\subsection{Proof of Theorem \ref{T2}}
The result of Lemma \ref{joint2} can be rewritten as
\begin{eqnarray*}
\left(\frac{1}{n}\mathbf{\Sigma}_2\right)^{-1/2} \cdot \sqrt{n}\left( \begin{array}{ccc} \frac{1}{n}\sum_{i=1}^n \hat{\varepsilon}_i^4 & - & \frac{3c_n(k+2)}{(n+2)}\\ \frac{1}{n}\sum_{i=1}^n \hat{\varepsilon}_i^2 & - & c_n\end{array}\right) \overset{\mathcal{D}}{\longrightarrow} \mathcal{N}\left( \mathbf{0}, \mathbf{I}_2\right).
\end{eqnarray*}
Due to the statistic $T_2$ can be rewritten as
\begin{eqnarray*}
T_2=\frac{\frac{1}{n}\sum_{i=1}^n \hat{\varepsilon}_i^4}{\left(\frac{1}{n}\sum_{i=1}^n \hat{\varepsilon}_i^2\right)^2}-1,
\end{eqnarray*}
define a function $f(x, y)=\frac{x}{y^2}-1$, then $T_2 = f(n^{-1}\sum_{i=1}^n \hat{\varepsilon}_i^4, n^{-1}\sum_{i=1}^n \hat{\varepsilon}_i^2)$. Let $\theta_1=\frac{3c_n(k+2)}{(n+2)}$ and $\theta_2=c_n$. Using delta method, $T_2$ is asymptotic Gaussian and we can obtain its limiting expectation and variance as follows. The expectation is
\begin{eqnarray*}
E(T_2)=f(\theta_1, \theta_2) = \frac{2+6/k -2/n}{1+2/n},
\end{eqnarray*}
and
\begin{eqnarray}
\lim_{k, n \to \infty}E(T_2)\to 2.
\end{eqnarray}
And the variance of $T_2$ is
\begin{eqnarray*}
var(T_2)=\nabla f \cdot \left(\frac{1}{n}\mathbf{\Sigma}_2\right) \nabla f^\prime,
\end{eqnarray*}
where $\nabla f=\left( f_x^\prime (\theta_1, \theta_2) \quad  f_y^\prime (\theta_1, \theta_2)\right)$
is the first order differential vector with $f_x^\prime (\theta_1, \theta_2) = \frac{1}{c_n^2}$, $f_y^\prime (\theta_1, \theta_2) = -6 \frac{(k+2)}{c_n^2(n+2)}$.
Finally, the variance is
\begin{eqnarray*}
var(T_2)= 24 + \frac{288}{k}+\frac{360}{c_nk}+O\left(\frac{1}{k^2}\right)+O\left(\frac{1}{n^2}\right),
\end{eqnarray*}
and
\begin{eqnarray}
\lim_{k, n\to \infty}var(T_2)\to 24.
\end{eqnarray}
The proof of Theorem \ref{T2} is complete.

\end{document}